\documentclass[aps,twocolumn,superscriptaddress,groupedaddress]{revtex4}

\newcommand{\bra}[1]{\langle #1|} 
\newcommand{\ket}[1]{|#1\rangle}
\newcommand{\braket}[2]{\langle #1|#2\rangle}

\newcommand{\Ora}[1]{\mathcal{O}_{#1}}
\newcommand{\OInt}{h}
\newcommand{\Norm}[1]{\left\lVert #1 \right\rVert}

\usepackage{amsmath,amsthm,amsfonts,amssymb}
\usepackage{fontenc}
\usepackage[all]{xy}
\usepackage{graphicx,subfigure,multirow}
\usepackage{tikz}
\usetikzlibrary{shapes, arrows, shadows}
\usepackage[scaled]{helvet}
\usepackage[toc, page]{appendix}
\usepackage{graphicx}				
\usepackage{amssymb}
\usepackage{dsfont}
\usepackage[colorlinks=true]{hyperref}
\newtheorem{definition}{Definition}
\newtheorem{theorem}{Theorem}
\newtheorem{lemma}{Lemma}

\newtheorem{principle}{Principle}
\newtheorem*{search problem}{Search Problem}

\usepackage{tikz}
\pgfdeclarelayer{nodelayer}
\pgfdeclarelayer{edgelayer}
\pgfsetlayers{edgelayer,nodelayer,main}
\tikzstyle{none}=[inner sep=0pt]
\tikzstyle{new}=[inner sep=2pt]
\usetikzlibrary{shapes, arrows, shadows}
\usetikzlibrary{circuits.ee.IEC}
\usetikzlibrary{decorations.pathmorphing}
\usetikzlibrary{shapes.geometric}

\tikzstyle{env}=[copoint,regular polygon rotate=0,minimum width=0.2cm, fill=black]

\tikzstyle{probs}=[shape=semicircle,fill=white,draw=black,shape border rotate=180,minimum width=1.2cm]

\usetikzlibrary{circuits.ee.IEC}
\usetikzlibrary{decorations.pathmorphing}
\usetikzlibrary{shapes.geometric}

\tikzstyle{wavy}=[decorate,decoration={snake, segment length=1mm, amplitude=0.1mm}]
\tikzstyle{mopoint}=[shape=semicircle, fill=white,draw=black,shape border rotate=180,scale =0.75]
\tikzstyle{mocopoint}=[shape=semicircle, fill=white,draw=black,minimum width = 0.9cm, scale =0.75, xscale=0.7]
\tikzstyle{cpoint}=[shape=semicircle, fill=white,draw=black,minimum width = 0.9cm, scale =0.75, xscale=1, yscale=0.7, shape border rotate = 90,font=\fontsize{14}{16}\selectfont]
\tikzstyle{cocpoint}=[shape=semicircle, fill=white,draw=black,minimum width = 0.9cm, scale =0.75, xscale=1, yscale=0.7, shape border rotate = 270,font=\fontsize{14}{16}\selectfont]
\tikzstyle{every picture}=[baseline=-0.25em,scale=0.5]
\tikzstyle{dotpic}=[] % for backwards-compatibility
\tikzstyle{diredges}=[every to/.style={diredge}]
\tikzstyle{math matrix}=[matrix of math nodes,left delimiter=(,right delimiter=),inner sep=2pt,column sep=1em,row sep=0.5em,nodes={inner sep=0pt},text height=1.5ex, text depth=0.25ex]

\tikzstyle{inline text}=[text height=1.5ex, text depth=0.25ex,yshift=0.5mm]
\tikzstyle{label}=[font=\footnotesize,text height=1.5ex, text depth=0.25ex,yshift=0.5mm]
\tikzstyle{left label}=[label,anchor=east,xshift=1.5mm]
\tikzstyle{right label}=[label,anchor=west,xshift=-1.5mm]

\tikzstyle{braceedge}=[decorate,decoration={brace,amplitude=2mm,raise=-1mm}]
\tikzstyle{small braceedge}=[decorate,decoration={brace,amplitude=1mm,raise=-1mm}]

\tikzstyle{doubled}=[line width=2pt] % set the line width for all doubled (quantum) maps/wires
\tikzstyle{boldedge}=[doubled,shorten <=-0.17mm,shorten >=-0.17mm]

\tikzstyle{boldedgedashed}=[very thick,dashed,shorten <=-0.17mm,shorten >=-0.17mm]
\tikzstyle{vboldedgedashed}=[doubled,dashed,shorten <=-0.17mm,shorten >=-0.17mm]
\tikzstyle{left hook arrow}=[left hook-latex]
\tikzstyle{right hook arrow}=[right hook-latex]
\tikzstyle{sembracket}=[line width=0.5pt,shorten <=-0.07mm,shorten >=-0.07mm]

\tikzstyle{causal edge}=[->,thick,gray]
\tikzstyle{causal nondir}=[thick,gray]
\tikzstyle{timeline}=[thick,gray, dashed]

\tikzstyle{cedge}=[<->,thick,gray!70!white]

\tikzstyle{empty diagram}=[draw=gray!40!white,dashed,shape=rectangle,minimum width=1cm,minimum height=1cm]
\tikzstyle{empty diagram small}=[draw=gray!50!white,dashed,shape=rectangle,minimum width=0.6cm,minimum height=0.5cm]

\tikzstyle{dot}=[inner sep=0.7mm,minimum width=0pt,minimum height=0pt,draw,shape=circle]
\tikzstyle{ddot}=[inner sep=0.7mm,doubled, minimum width=2.5mm,minimum height=2.5mm,draw,shape=circle]

\tikzstyle{black dot}=[dot,fill=black]
\tikzstyle{white dot}=[dot,fill=white]
\tikzstyle{green dot}=[white dot] % for backwards-compatibility
\tikzstyle{gray dot}=[dot,fill=gray!40!white]
\tikzstyle{red dot}=[gray dot] % for backwards-compatibility

\tikzstyle{black ddot}=[ddot,fill=black]
\tikzstyle{white ddot}=[ddot,fill=white]
\tikzstyle{gray ddot}=[ddot,fill=gray!40!white]

\tikzstyle{gray edge}=[gray!40!white]

\tikzstyle{small dot}=[inner sep=0.4mm,minimum width=0pt,minimum height=0pt,draw,shape=circle]

\tikzstyle{small black dot}=[small dot,fill=black]
\tikzstyle{small white dot}=[small dot,fill=white]
\tikzstyle{small gray dot}=[small dot,fill=gray!40!white]

\tikzstyle{causal dot}=[inner sep=0.4mm,minimum width=0pt,minimum height=0pt,draw=white,shape=circle,fill=gray!40!white]

\tikzstyle{white phase dot}=[dot,fill=white]
\tikzstyle{white phase ddot}=[ddot,fill=white]

\tikzstyle{gray phase dot}=[dot,fill=gray!40!white]
\tikzstyle{gray phase ddot}=[ddot,fill=gray!40!white]
\tikzstyle{grey phase dot}=[gray phase dot]
\tikzstyle{grey phase ddot}=[gray phase ddot]

\tikzstyle{cnot}=[fill=white,shape=circle,inner sep=-1.4pt]
\tikzstyle{hadamard}=[square box,inner sep=0 pt,font=\tiny\sf,minimum height=3mm,minimum width=3mm]
\tikzstyle{dhadamard}=[hadamard,doubled]
\tikzstyle{antipode}=[white dot,inner sep=0.3mm,font=\footnotesize]

\tikzstyle{scalar}=[diamond,draw,inner sep=0.5pt,font=\small]
\tikzstyle{dscalar}=[diamond,doubled, draw,inner sep=0.5pt,font=\small]

\tikzstyle{small box}=[rectangle,inline text,fill=white,draw,minimum height=5mm,yshift=-0.5mm,minimum width=5mm,font=\small]
\tikzstyle{small gray box}=[small box,fill=gray!30]
\tikzstyle{medium box}=[rectangle,inline text,fill=white,draw,minimum height=5mm,yshift=-0.5mm,minimum width=10mm,font=\small]
\tikzstyle{square box}=[small box] % for backwards-compatibility
\tikzstyle{medium gray box}=[small box,fill=gray!30]
\tikzstyle{large box}=[rectangle,inline text,fill=white,draw,minimum height=5mm,yshift=-0.5mm,minimum width=15mm,font=\small]
\tikzstyle{large gray box}=[small box,fill=gray!30]

\tikzstyle{point}=[regular polygon,regular polygon sides=3,draw,scale=0.75,inner sep=-0.5pt,minimum width=9mm,fill=white,regular polygon rotate=180]
\tikzstyle{copoint}=[regular polygon,regular polygon sides=3,draw,scale=0.75,inner sep=-0.5pt,minimum width=9mm,fill=white]
\tikzstyle{dpoint}=[point,doubled]
\tikzstyle{dcopoint}=[copoint,doubled]

\tikzstyle{tinypoint}=[regular polygon,regular polygon sides=3,draw,scale=0.55,inner sep=-0.15pt,minimum width=6mm,fill=white,regular polygon rotate=180]
\tikzstyle{arrowhead}=[regular polygon,regular polygon sides=3,draw,scale=0.2,inner sep=-0.15pt,minimum width=6mm,fill=black,regular polygon rotate=180]

\tikzstyle{white point}=[point]
\tikzstyle{green point}=[white point] % for backwards-compatibility
\tikzstyle{white copoint}=[copoint]
\tikzstyle{gray point}=[point,fill=gray!40!white]
\tikzstyle{gray dpoint}=[gray point,doubled]
\tikzstyle{red point}=[gray point] % for backwards-compatibility
\tikzstyle{gray copoint}=[copoint,fill=gray!40!white]
\tikzstyle{gray dcopoint}=[gray copoint,doubled]

\tikzstyle{tiny gray point}=[tinypoint,fill=gray!40!white]

\tikzstyle{diredge}=[->]
\tikzstyle{rdiredge}=[<-]
\tikzstyle{thickdiredge}=[->, very thick]
\tikzstyle{pointer edge}=[->,very thick,gray]
\tikzstyle{pointer edge part}=[very thick,gray]
\tikzstyle{dashed edge}=[dashed]
\tikzstyle{thick dashed edge}=[very thick,dashed]
\tikzstyle{thick gray dashed edge}=[thick dashed edge,gray!90]
\tikzstyle{thick map edge}=[very thick,|->]

\makeatletter
\newcommand{\boxshape}[3]{%
\pgfdeclareshape{#1}{
\inheritsavedanchors[from=rectangle] % this is nearly a rectangle
\inheritanchorborder[from=rectangle]
\inheritanchor[from=rectangle]{center}
\inheritanchor[from=rectangle]{north}
\inheritanchor[from=rectangle]{south}
\inheritanchor[from=rectangle]{west}
\inheritanchor[from=rectangle]{east}
\backgroundpath{% this is new
\southwest \pgf@xa=\pgf@x \pgf@ya=\pgf@y
\northeast \pgf@xb=\pgf@x \pgf@yb=\pgf@y

\@tempdima=#2
\@tempdimb=#3

\pgfpathmoveto{\pgfpoint{\pgf@xa - 5pt + \@tempdima}{\pgf@ya}}
\pgfpathlineto{\pgfpoint{\pgf@xa - 5pt - \@tempdima}{\pgf@yb}}
\pgfpathlineto{\pgfpoint{\pgf@xb + 5pt + \@tempdimb}{\pgf@yb}}
\pgfpathlineto{\pgfpoint{\pgf@xb + 5pt - \@tempdimb}{\pgf@ya}}
\pgfpathlineto{\pgfpoint{\pgf@xa - 5pt + \@tempdima}{\pgf@ya}}
\pgfpathclose
}
}}

\boxshape{NEbox}{0pt}{5pt}
\boxshape{SEbox}{0pt}{-5pt}
\boxshape{NWbox}{5pt}{0pt}
\boxshape{SWbox}{-5pt}{0pt}
\boxshape{EBox}{-3pt}{3pt}
\boxshape{WBox}{3pt}{-3pt}
\makeatother

\tikzstyle{cloud}=[shape=cloud,draw,minimum width=1.5cm,minimum height=1.5cm]

\tikzstyle{map}=[draw,shape=NEbox,inner sep=2pt,minimum height=6mm,fill=white]
\tikzstyle{mapdag}=[draw,shape=SEbox,inner sep=2pt,minimum height=6mm,fill=white]
\tikzstyle{mapadj}=[draw,shape=SEbox,inner sep=2pt,minimum height=6mm,fill=white]
\tikzstyle{maptrans}=[draw,shape=SWbox,inner sep=2pt,minimum height=6mm,fill=white]
\tikzstyle{mapconj}=[draw,shape=NWbox,inner sep=2pt,minimum height=6mm,fill=white]

\tikzstyle{dbox}=[draw,doubled,shape=rectangle,inner sep=2pt,minimum height=6mm,minimum width=6mm,fill=white]
\tikzstyle{dmap}=[draw,doubled,shape=NEbox,inner sep=2pt,minimum height=6mm,fill=white]
\tikzstyle{dmapdag}=[draw,doubled,shape=SEbox,inner sep=2pt,minimum height=6mm,fill=white]
\tikzstyle{dmapadj}=[draw,doubled,shape=SEbox,inner sep=2pt,minimum height=6mm,fill=white]
\tikzstyle{dmaptrans}=[draw,doubled,shape=SWbox,inner sep=2pt,minimum height=6mm,fill=white]
\tikzstyle{dmapconj}=[draw,doubled,shape=NWbox,inner sep=2pt,minimum height=6mm,fill=white]

\tikzstyle{ddmap}=[draw,doubled,dashed,shape=NEbox,inner sep=2pt,minimum height=6mm,fill=white]
\tikzstyle{ddmapdag}=[draw,doubled,dashed,shape=SEbox,inner sep=2pt,minimum height=6mm,fill=white]
\tikzstyle{ddmapadj}=[draw,doubled,dashed,shape=SEbox,inner sep=2pt,minimum height=6mm,fill=white]
\tikzstyle{ddmaptrans}=[draw,doubled,dashed,shape=SWbox,inner sep=2pt,minimum height=6mm,fill=white]
\tikzstyle{ddmapconj}=[draw,doubled,dashed,shape=NWbox,inner sep=2pt,minimum height=6mm,fill=white]

\boxshape{sNEbox}{0pt}{3pt}
\boxshape{sSEbox}{0pt}{-3pt}
\boxshape{sNWbox}{3pt}{0pt}
\boxshape{sSWbox}{-3pt}{0pt}
\tikzstyle{smap}=[draw,shape=sNEbox,fill=white]
\tikzstyle{smapdag}=[draw,shape=sSEbox,fill=white]
\tikzstyle{smapadj}=[draw,shape=sSEbox,fill=white]
\tikzstyle{smaptrans}=[draw,shape=sSWbox,fill=white]
\tikzstyle{smapconj}=[draw,shape=sNWbox,fill=white]

\tikzstyle{dsmap}=[draw,dashed,shape=sNEbox,fill=white]
\tikzstyle{dsmapdag}=[draw,dashed,shape=sSEbox,fill=white]
\tikzstyle{dsmaptrans}=[draw,dashed,shape=sSWbox,fill=white]
\tikzstyle{dsmapconj}=[draw,dashed,shape=sNWbox,fill=white]

\boxshape{mNEbox}{0pt}{10pt}
\boxshape{mSEbox}{0pt}{-10pt}
\boxshape{mNWbox}{10pt}{0pt}
\boxshape{mSWbox}{-10pt}{0pt}
\tikzstyle{mmap}=[draw,shape=mNEbox]
\tikzstyle{mmapdag}=[draw,shape=mSEbox]
\tikzstyle{mmaptrans}=[draw,shape=mSWbox]
\tikzstyle{mmapconj}=[draw,shape=mNWbox]

\tikzstyle{mmapgray}=[draw,fill=gray!40!white,shape=mNEbox]
\tikzstyle{smapgray}=[draw,fill=gray!40!white,shape=sNEbox]

\makeatletter
\pgfdeclareshape{cornerpoint}{
\inheritsavedanchors[from=rectangle] % this is nearly a rectangle
\inheritanchorborder[from=rectangle]
\inheritanchor[from=rectangle]{center}
\inheritanchor[from=rectangle]{north}
\inheritanchor[from=rectangle]{south}
\inheritanchor[from=rectangle]{west}
\inheritanchor[from=rectangle]{east}
\backgroundpath{% this is new
\southwest \pgf@xa=\pgf@x \pgf@ya=\pgf@y
\northeast \pgf@xb=\pgf@x \pgf@yb=\pgf@y

\pgfmathsetmacro{\pgf@shorten@left}{\pgfkeysvalueof{/tikz/shorten left}}
\pgfmathsetmacro{\pgf@shorten@right}{\pgfkeysvalueof{/tikz/shorten right}}

\pgfpathmoveto{\pgfpoint{0.5 * (\pgf@xa + \pgf@xb)}{\pgf@ya - 5pt}}
\pgfpathlineto{\pgfpoint{\pgf@xa - 8pt + \pgf@shorten@left}{\pgf@yb - 1.5 * \pgf@shorten@left}}
\pgfpathlineto{\pgfpoint{\pgf@xa - 8pt + \pgf@shorten@left}{\pgf@yb}}
\pgfpathlineto{\pgfpoint{\pgf@xb + 8pt - \pgf@shorten@right}{\pgf@yb}}
\pgfpathlineto{\pgfpoint{\pgf@xb + 8pt - \pgf@shorten@right}{\pgf@yb - 1.5 * \pgf@shorten@right}}
\pgfpathclose
}
}

\pgfdeclareshape{cornercopoint}{
\inheritsavedanchors[from=rectangle] % this is nearly a rectangle
\inheritanchorborder[from=rectangle]
\inheritanchor[from=rectangle]{center}
\inheritanchor[from=rectangle]{north}
\inheritanchor[from=rectangle]{south}
\inheritanchor[from=rectangle]{west}
\inheritanchor[from=rectangle]{east}
\backgroundpath{% this is new
\southwest \pgf@xa=\pgf@x \pgf@ya=\pgf@y
\northeast \pgf@xb=\pgf@x \pgf@yb=\pgf@y

\pgfmathsetmacro{\pgf@shorten@left}{\pgfkeysvalueof{/tikz/shorten left}}
\pgfmathsetmacro{\pgf@shorten@right}{\pgfkeysvalueof{/tikz/shorten right}}

\pgfpathmoveto{\pgfpoint{0.5 * (\pgf@xa + \pgf@xb)}{\pgf@yb + 5pt}}
\pgfpathlineto{\pgfpoint{\pgf@xa - 8pt + \pgf@shorten@left}{\pgf@ya + 1.5 * \pgf@shorten@left}}
\pgfpathlineto{\pgfpoint{\pgf@xa - 8pt + \pgf@shorten@left}{\pgf@ya}}
\pgfpathlineto{\pgfpoint{\pgf@xb + 8pt - \pgf@shorten@right}{\pgf@ya}}
\pgfpathlineto{\pgfpoint{\pgf@xb + 8pt - \pgf@shorten@right}{\pgf@ya + 1.5 * \pgf@shorten@right}}
\pgfpathclose
}
}

\makeatother

\pgfkeyssetvalue{/tikz/shorten left}{0pt}
\pgfkeyssetvalue{/tikz/shorten right}{0pt}

\tikzstyle{kpoint common}=[draw,fill=white,inner sep=1pt,minimum height=4mm]
\tikzstyle{kpoint}=[shape=cornerpoint,shorten left=5pt,kpoint common]
\tikzstyle{kpoint adjoint}=[shape=cornercopoint,shorten left=5pt,kpoint common]
\tikzstyle{kpoint conjugate}=[shape=cornerpoint,shorten right=5pt,kpoint common]
\tikzstyle{kpoint transpose}=[shape=cornercopoint,shorten right=5pt,kpoint common]
\tikzstyle{kpoint symm}=[shape=cornerpoint,shorten left=5pt,shorten right=5pt,kpoint common]

\tikzstyle{kpointdag}=[kpoint adjoint]
\tikzstyle{kpointadj}=[kpoint adjoint]
\tikzstyle{kpointconj}=[kpoint conjugate]
\tikzstyle{kpointtrans}=[kpoint transpose]

\tikzstyle{dkpoint}=[kpoint,doubled]
\tikzstyle{dkpointdag}=[kpoint adjoint,doubled]
\tikzstyle{dkcopoint}=[kpoint adjoint,doubled]
\tikzstyle{dkpointadj}=[kpoint adjoint,doubled]
\tikzstyle{dkpointconj}=[kpoint conjugate,doubled]
\tikzstyle{dkpointtrans}=[kpoint transpose,doubled]

\tikzstyle{kscalar}=[kpoint common, shape=EBox, inner xsep=-1pt, inner ysep=3pt,font=\small]
\tikzstyle{kscalarconj}=[kpoint common, shape=WBox, inner xsep=-1pt, inner ysep=3pt,font=\small]

 \tikzstyle{upground}=[circuit ee IEC,thick,ground,rotate=90,scale=2.5]
 \tikzstyle{downground}=[circuit ee IEC,thick,ground,rotate=-90,scale=2.5]
 \tikzstyle{bigground}=[regular polygon,regular polygon sides=3,draw=gray,scale=0.50,inner sep=-0.5pt,minimum width=10mm,fill=gray]
\tikzstyle{arrs}=[-latex,font=\small,auto]
\tikzstyle{arrow plain}=[arrs]
\tikzstyle{arrow dashed}=[dashed,arrs]
\tikzstyle{arrow bold}=[very thick,arrs]
\tikzstyle{arrow hide}=[draw=white!0,-]
\tikzstyle{arrow reverse}=[latex-]
\tikzstyle{cdnode}=[]

\tikzstyle{slit}=[line width=2]
\tikzstyle{block}=[line width=4,black!10!red,line cap=round]
\tikzstyle{screen}=[line width=2,black!30!gray,line cap=round]
\tikzstyle{di}=[diamond,draw,inner sep=0.5pt,font=\small, minimum size = .5cm]
\tikzstyle{sbox}=[rectangle,draw]
\tikzstyle{mirror}=[line width=2,black]
\tikzstyle{trace}=[circuit ee IEC,thick,ground,rotate=0,scale=2]
\tikzstyle{traceState}=[circuit ee IEC,thick,ground,rotate=180,scale=2]
\tikzstyle{detEff}=[circuit ee IEC,thick,ground,rotate=180,scale=1.4]
\tikzstyle{maxMix}=[circuit ee IEC,thick,ground,scale=1.4]
\tikzstyle{particlePath}=[line width=2,gray!40, line cap =round]
\usepackage{fp}
\usetikzlibrary{calc,decorations.pathmorphing,decorations.text,fixedpointarithmetic}

\begin{document}

\title{Deriving Grover's lower bound from simple physical principles}
\author{Ciar{\'a}n~M. Lee}
\email{ciaran.lee@cs.ox.ac.uk}
\affiliation{University of Oxford, Department of Computer Science, Wolfson Building, Parks Road, Oxford OX1 3QD, UK.}
\author{John~H. Selby}
\email{john.selby08@ic.ac.uk}
\affiliation{University of Oxford, Department of Computer Science, Wolfson Building, Parks Road, Oxford OX1 3QD, UK.}
\affiliation{Imperial College London,  London SW7 2AZ, UK.}

\begin{abstract}
Grover's algorithm constitutes the optimal quantum solution to the search problem and provides a quadratic speed-up over all possible classical search algorithms. Quantum interference between computational paths has been posited as a key resource behind this computational speed-up. However there is a limit to this interference, at most pairs of paths can ever interact in a fundamental way. Could more interference imply more computational power? Sorkin has defined a hierarchy of possible interference behaviours---currently under experimental investigation---where classical theory is at the first level of the hierarchy and quantum theory belongs to the second. Informally, the order in the hierarchy corresponds to the number of paths that have an irreducible interaction in a multi-slit experiment. In this work, we consider how Grover's speed-up depends on the order of interference in a theory. Surprisingly, we show that the quadratic lower bound holds regardless of the order of interference. Thus, at least from the point of view of the search problem, post-quantum interference does not imply a computational speed-up over quantum theory.

\end{abstract}

\maketitle

Grover's algorithm \cite{grover1997quantum} provides the optimal quantum solution to the search problem and is one of the most versatile and influential quantum algorithms. The search problem---in its simplest form---asks one to find a single ``marked'' item from an unstructured list of $N$ elements by querying an oracle which can recognise the marked item. The importance of Grover's algorithm stems from the ubiquitous nature of the search problem and its relation to solving \textbf{NP}-complete problems \cite{bennett1997strengths}. Classical computers require $O(N)$ queries to solve this problem, but quantum computers---using Grover's algorithm---only require $O(\sqrt{N})$ queries. Quantum interference between computational paths has been posited \cite{stahlke2014quantum} as a key resource behind this computational ``speed-up''. However, as first noted by Sorkin \cite{sorkin1994quantum,sorkin1995quantum}, there is a limit to this interference---at most pairs of paths can ever interact in a fundamental way. Could more interference imply more computational power?

Sorkin has defined a hierarchy of possible interference behaviours---currently under experimental investigation \cite{sinha2008testing,park2012three,sinha2010ruling}---where classical theory is at the first level of the hierarchy and quantum theory belongs to the second. Informally, the order in the hierarchy corresponds to the number of paths that have an irreducible interaction in a multi-slit experiment. To get a greater understanding of the role of interference in computation, we consider how Grover's speed-up depends on the order of interference in a theory.

Restriction to the second level of this hierarchy implies many ``quantum-like'' features, which, at first glance, appear to be unrelated to interference. For example, such interference behaviour restricts correlations \cite{dowker2014histories} to the ``almost quantum correlations'' discussed in \cite{navascues2015almost}, and bounds contextuality in a manner similar to quantum theory \cite{henson2015bounding,niestegge2012conditional}. This, in conjunction with interference being a key resource in the quantum speed-up, suggests that post-quantum interference may allow for a speed-up over quantum computation.

Surprisingly, we show that this is not the case---at least from the point of view of the search problem. We consider this problem within the framework of generalised probabilistic theories, which is suitable for describing arbitrary operationally-defined theories \cite{Pavia1,Pavia2,Hardy2011,lee2015computation,lee2015proofs,barrett2015landscape}. Classical probability theory, quantum theory, Spekken's toy model \cite{janotta2013generalized,spekkens2007evidence}, and the theory of PR boxes \cite{popescu1998causality} all provide examples of theories in this framework. We consider theories satisfying certain natural physical principles which are sufficient for the existence of a well-defined search oracle. Given these physical principles, we prove that a theory at level $\OInt$ in Sorkin's hierarchy requires $\Omega(\sqrt{N/\OInt})$ queries to solve the search problem. %Using this, we show that a theory at level $\OInt$ in Sorkin's hierarchy requires $\Omega(\sqrt{N/\OInt})$ queries to solve the search problem. %to this search oracle to find a single ``marked'' item from some $N$-element list.
Thus, post-quantum interference does not imply a computational speed-up over quantum theory. Moreover, from the point of view of the search problem, all (finite) orders of interference are asymptotically equivalent.

\section{Generalised probabilistic theories}
A basic requirement of any physical theory is that it should provide a consistent account of experimental data. This idea underlies the framework of generalised probabilistic theories---developed in \cite{Pavia1, Pavia2, Hardy2011, barrett2007information}---which allows for the description of arbitrary theories satisfying this requirement. Informally, a theory in this framework specifies a set of \emph{physical processes} which can be connected together to form experiments. Each process corresponds to a single use of a piece of laboratory apparatus, each having a number of input and output ports, as well as a classical pointer. When the physical apparatus is used in an experiment, the classical pointer comes to rest at one of a number of positions, indicating a specific outcome has occurred.
Each port is associated with a \emph{physical system} of a particular type (labelled $A,B,...$). Intuitively one can consider these physical systems as passing from outputs of one process to inputs of another. Processes can thus be connected together---both in sequence and in parallel---to form \emph{circuits}, where it is required that types match and there are no cycles.
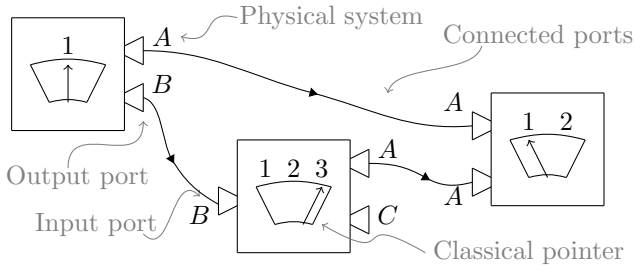
\[\begin{tikzpicture}
	\begin{pgfonlayer}{nodelayer}
		\node [style=none] (0) at (0.4999999, 0.7499999) {};
		\node [style=none] (1) at (0.9999999, 0.4999999) {};
		\node [style=none] (2) at (0.9999999, 0.25) {};
		\node [style=none] (3) at (0.4999999, -0) {};
		\node [style=none] (4) at (1, 2) {};
		\node [style=none] (5) at (4, 2) {};
		\node [style=none] (6) at (4, -1) {};
		\node [style=none] (7) at (1, -1) {};
		\node [style=none] (8) at (2, -0.25) {};
		\node [style=none] (9) at (3, -0.25) {};
		\node [style=none] (10) at (1.5, 0.75) {};
		\node [style=none] (11) at (3.5, 0.75) {};
		\node [style=none] (12) at (2.75, -0.25) {};
		\node [style=none] (13) at (3.25, 0.7500002) {};
		\node [style=none] (14) at (1.75, 1.25) {$1$};
		\node [style=none] (15) at (2.5, 1.25) {$2$};
		\node [style=none] (16) at (3.25, 1.25) {$3$};
		\node [style=none] (17) at (4.5, -0.4999999) {};
		\node [style=none] (18) at (4, -0.25) {};
		\node [style=none] (19) at (4.5, 0.25) {};
		\node [style=none] (20) at (4, -0) {};
		\node [style=none] (21) at (8.75, 2) {};
		\node [style=none] (22) at (9.75, 2.5) {$2$};
		\node [style=none] (23) at (7.25, 2.75) {};
		\node [style=none] (24) at (7.75, 3.25) {};
		\node [style=none] (25) at (7.75, 2.5) {};
		\node [style=none] (26) at (10.75, 0.25) {};
		\node [style=none] (27) at (7.75, 0.9999999) {};
		\node [style=none] (28) at (9.25, 1) {};
		\node [style=none] (29) at (7.75, 0.7499999) {};
		\node [style=none] (30) at (7.75, 2.25) {};
		\node [style=none] (31) at (10.25, 2) {};
		\node [style=none] (32) at (7.25, 2) {};
		\node [style=none] (33) at (7.25, 0.4999999) {};
		\node [style=none] (34) at (9.75, 0.9999999) {};
		\node [style=none] (35) at (8.25, 2) {};
		\node [style=none] (36) at (8.75, 2.5) {$1$};
		\node [style=none] (37) at (7.75, 0.25) {};
		\node [style=none] (38) at (10.75, 3.25) {};
		\node [style=none] (39) at (7.25, 1.25) {};
		\node [style=none] (40) at (8.75, 0.9999999) {};
		\node [style=none] (41) at (-3.5, 4) {};
		\node [style=none] (42) at (-5, 5.25) {};
		\node [style=none] (43) at (-2, 2.25) {};
		\node [style=none] (44) at (-1.5, 4) {};
		\node [style=none] (45) at (-3.5, 3) {};
		\node [style=none] (46) at (-2, 4.5) {};
		\node [style=none] (47) at (-1.5, 4.75) {};
		\node [style=none] (48) at (-2.5, 4) {};
		\node [style=none] (49) at (-2, 4.25) {};
		\node [style=none] (50) at (-3, 3) {};
		\node [style=none] (51) at (-4.5, 4) {};
		\node [style=none] (52) at (-3.5, 4.5) {$1$};
		\node [style=none] (53) at (-5, 2.25) {};
		\node [style=none] (54) at (-2, 5.25) {};
		\node [style=none] (55) at (-4, 3) {};
		\node [style=none] (56) at (4, 1.25) {};
		\node [style=none] (57) at (4, 1.5) {};
		\node [style=none] (58) at (4.5, 1.75) {};
		\node [style=none] (59) at (4.5, 0.9999999) {};
		\node [style=none] (60) at (-2, 3) {};
		\node [style=none] (61) at (-2, 3.25) {};
		\node [style=none] (62) at (-1.5, 3.5) {};
		\node [style=none] (63) at (-1.5, 2.75) {};
		\node [style=none] (64) at (-1.5, 4.375) {};
		\node [style=none] (65) at (-1.5, 3.125) {};
		\node [style=arrowhead,rotate=28] (66) at (-0.75, 1.5) {};
		\node [style=none] (67) at (4.5, 1.375) {};
		\node [style=none] (68) at (4.5, -0.125) {};
		\node [style=arrowhead,rotate=195] (69) at (3, 3.25) {};
		\node [style=arrowhead,rotate=180] (70) at (6, 1) {};
		\node [style=none] (71) at (-0.9999999, 4.75) {$A$};
		\node [style=none] (72) at (6.75, 3) {$A$};
		\node [style=none] (73) at (-0.9999999, 3.5) {$B$};
		\node [style=none] (74) at (0, -0) {$B$};
		\node [style=none] (75) at (5, 1.75) {$A$};
		\node [style=none] (76) at (6.75, 0.5000001) {$A$};
		\node [style=none] (77) at (5, -0) {$C$};
		\node [style=none, color=gray] (78) at (3.5, 5.25) {Physical system};
		\node [style=none, color=gray] (79) at (9.000001, 4.75) {Connected ports};
		\node [style=none] (80) at (0.9999998, 5.25) {};
		\node [style=none] (81) at (-0.4999998, 4.75) {};
		\node [style=none] (82) at (5, 3.25) {};
		\node [style=none] (83) at (7.5, 4.5) {};
		\node [style=none] (84) at (3.25, -0.2500001) {};
		\node [style=none] (85) at (6, -0.9999999) {};
		\node [style=none] (86) at (-1.5, 2.5) {};
		\node [style=none] (87) at (0, 0.5000001) {};
		\node [style=none] (88) at (-3.5, 1.25) {};
		\node [style=none] (89) at (-1.5, -0.4999999) {};
		\node [style=none, color=gray] (90) at (8.75, -0.9999999) {Classical pointer};
		\node [style=none, color=gray] (91) at (-3.25, 0.9999998) {Output port};
		\node [style=none, color=gray] (92) at (-2.75, -0.2500001) {Input port};
		\node [style=none] (93) at (0.5000002, 0.2875) {};
		\node [style=none] (94) at (7.25, 2.375) {};
		\node [style=none] (95) at (7.25, 0.875) {};
	\end{pgfonlayer}
	\begin{pgfonlayer}{edgelayer}
		\draw (0.center) to (1.center);
		\draw (3.center) to (2.center);
		\draw (4.center) to (7.center);
		\draw (7.center) to (6.center);
		\draw (6.center) to (5.center);
		\draw (5.center) to (4.center);
		\draw (10.center) to (8.center);
		\draw (9.center) to (11.center);
		\draw [bend left=15, looseness=1.00] (8.center) to (9.center);
		\draw [bend left=15, looseness=1.00] (10.center) to (11.center);
		\draw [style=diredge] (12.center) to (13.center);
		\draw (19.center) to (20.center);
		\draw (17.center) to (18.center);
		\draw [style=none] (0.center) to (3.center);
		\draw [style=none] (19.center) to (17.center);
		\draw (23.center) to (25.center);
		\draw (32.center) to (30.center);
		\draw (24.center) to (37.center);
		\draw (37.center) to (26.center);
		\draw (26.center) to (38.center);
		\draw (38.center) to (24.center);
		\draw (35.center) to (40.center);
		\draw (34.center) to (31.center);
		\draw [bend left=15, looseness=1.00] (40.center) to (34.center);
		\draw [bend left=15, looseness=1.00] (35.center) to (31.center);
		\draw [style=diredge] (28.center) to (21.center);
		\draw (39.center) to (27.center);
		\draw (33.center) to (29.center);
		\draw [style=none] (23.center) to (32.center);
		\draw [style=none] (39.center) to (33.center);
		\draw (42.center) to (53.center);
		\draw (53.center) to (43.center);
		\draw (43.center) to (54.center);
		\draw (54.center) to (42.center);
		\draw (51.center) to (55.center);
		\draw (50.center) to (48.center);
		\draw [bend left=15, looseness=1.00] (55.center) to (50.center);
		\draw [bend left=15, looseness=1.00] (51.center) to (48.center);
		\draw [style=diredge] (45.center) to (41.center);
		\draw (47.center) to (46.center);
		\draw (44.center) to (49.center);
		\draw [style=none] (47.center) to (44.center);
		\draw (58.center) to (57.center);
		\draw (59.center) to (56.center);
		\draw [style=none] (58.center) to (59.center);
		\draw (62.center) to (61.center);
		\draw (63.center) to (60.center);
		\draw [style=none] (62.center) to (63.center);
		\draw [color=black, in=165, out=0, looseness=1.00] (64.center) to (69.center);
		\draw [color=black, in=120, out=-15, looseness=0.75] (65.center) to (66.center);
		\draw [color=black, in=150, out=0, looseness=1.00] (67.center) to (70.center);
		\draw [style=diredge, color=gray, in=-15, out=165, looseness=1.75] (80.center) to (81.center);
		\draw [style=diredge, color=gray, in=75, out=-105, looseness=0.75] (83.center) to (82.center);
		\draw [style=diredge, color=gray, in=-30, out=150, looseness=2.25] (85.center) to (84.center);
		\draw [style=diredge, color=gray, in=135, out=-45, looseness=1.50] (89.center) to (87.center);
		\draw [style=diredge, color=gray, in=-90, out=90, looseness=1.50] (88.center) to (86.center);
		\draw [color=black, in=-175, out=-15, looseness=1.00] (69.center) to (94.center);
		\draw [color=black, bend right=15, looseness=1.00] (66.center) to (93.center);
		\draw [color=black, in=-165, out=-30, looseness=1.00] (70.center) to (95.center);
	\end{pgfonlayer}
\end{tikzpicture} \]
Closed circuits (i.e. circuits with no disconnected ports) correspond to the probability of obtaining a particular set of outcomes from the experiment represented by that circuit. Processes that yield the same probabilities in all closed circuits are identified, giving rise to equivalence classes of processes. Each element of such an equivalence class has the same input and output ports, and are denoted ${}_AT_B\in{}_A\mathcal{T}_B$, where ${}_A\mathcal{T}_B$ is the set of possible \emph{transformations} from systems $A$ to $B$. Transformations with no input ports are called \emph{states} $S_A \in \mathcal{S}_A$, and no output ports, \emph{effects}, ${}_AE\in{}_A\mathcal{E}$. %The symbol $\otimes$ will be used to denote the parallel composition between two transformations. Note that, for an arbitrary theory, $\otimes$ need not correspond to the standard vector space tensor product \cite{Pavia1}.

Given the probabilistic structure provided by closed circuits, each transformation ${}_AT_B$ can be associated with a real vector such that the set ${}_A\mathcal{T}_B$ is a subset of some real vector space, denoted ${}_AV_B$ \cite{Pavia1}. We assume in this work that all vector spaces are finite dimensional. It can be shown that transformations and effects act linearly on the vector space of states, $V_A$ \cite{Pavia1}. A measurement corresponds to a set of effects $\{e^r\}$ labelled by the position of the classical pointer $r$. The probability of preparing state $s$ and observing outcome $r$ is (suppressing system types for readability) given by: $$e^r(s) = P(r,s).$$

A state is \emph{pure} if it does not arise as a \emph{coarse-graining} of other states \footnote{The process $\{\mathcal{U}_j\}_{j\in{Y}}$, where $j$ indexes the classical pointer, is a coarse-graining of $\{\mathcal{E}_i\}_{i\in{X}}$ if there is a disjoint partition $\{X_j\}_{j\in{Y}}$ of $X$ such that $\mathcal{U}_j=\sum_{i\in{X_j}}\mathcal{E}_i$.}; a pure state is one for which we have maximal information. A state is \emph{mixed} if it is not pure. Similarly, one says a transformation is pure if it does not arise as a coarse-graining of other transformations. It can be shown that reversible transformations preserve pure states \cite{Pavia2}.

We now introduce five physical principles which will be assumed throughout the rest of this work. These can be though of as an abstraction of basic characteristics of the behaviour of information in quantum theory. Note however that these principles are not unique to quantum theory, indeed, real vector space quantum theory, fermionic quantum theory and the classical theory of pure states each satisfy all of these principles.
 
 \begin{principle} \textbf{Causality \cite{Pavia1}:}
There exists a unique deterministic effect ${}_AU$ for every system $A$, such that $\sum_r e^r=U $ for all measurements, $\{e^r\}_r$.
\end{principle}

In quantum theory the unique deterministic effect is provided by the partial trace. Mathematically, causality is equivalent to the statement: ``probabilities of present experiments are independent of future measurement choice'' \cite{Pavia1}, and so this can be interpreted as saying that ``information propagates from present to future''.
  
The deterministic effect allows one to define a notion of \emph{marginalisation} for multipartite states.

\begin{principle}\textbf{Purification \cite{Pavia1}:}
Given a state $s_{A}$ there exists a system $B$ and a pure state $S_{AB}$ on $AB$ such that $s_{A}$ is the marginalisation of $S_{AB}$: $${}_BU(S_{AB})=s_A.$$ Moreover, the purification $S_{AB}$ is unique up to reversible transformations on the purifying system, $B$ \footnote{Two states $S_{AB}$ and $S'_{AB}$ purifying $s_A$ satisfy $S_{AB}={\mathbb{I}\otimes_BT_B}(S'_{AB})$, with $_BT_B$ a reversible transformation.}.
\end{principle}

For example, in quantum theory any mixed state $\rho=\sum_i p_i \ket{\psi_i}\bra{\psi_i}$ can be written as $\rho= tr_B(\ket{\Psi}\bra{\Psi}_{AB})$ where $\ket{\Psi}_{AB}:=\sum_i\sqrt{p_i}\ket{\psi_i}\ket{i}$. Moreover, any other purification $\ket{\widetilde{\Psi}}_{AB}$ must satisfy $\ket{\Psi}_{AB}=\left(\mathbb{I}_A\otimes U_B\right)\ket{\widetilde{\Psi}}_{AB}$ with $U_B$ a unitary transformation. More generally this can be thought of as saying that information cannot be fundamentally destroyed, only discarded.

\begin{principle}\textbf{Purity Preservation \cite{chiribella2015operational}:}
The composite of pure transformations is pure.
\end{principle}

Pure transformations in quantum theory can be characterised by having Kraus rank $1$. Given two such transformations, their sequential or parallel composition will each also be rank $1$, and so composition preserves purity.

\begin{principle}\textbf{Pure Sharp Effect \cite{chiribella2015operational}:}
For each system $A$ there exists a pure effect that occurs with unit probability on some state.
\end{principle}

Pure states $\{a^i\}_{i=1}^n$ are \emph{perfectly distinguishable} if there exists a measurement, corresponding to effects $\{e^j\}_{j=1}^n$, such that $e^j(a^i)=\delta_{ij}$ for all $i,j$. For example, in quantum theory the computational basis $\{\ket{i}\}$ provide a perfectly distinguishable set, where the corresponding effects are just $\{\bra{j}\}$ such that $\braket{j}{i}=\delta_{ij}$. Such an $n$-tuple of states can reliably encode an $n$-level classical system.

\begin{principle}\textbf{Strong symmetry \cite{barnum2014higher}:}
For any two $n$-tuples of pure and perfectly distinguishable states $\{a^i\}$, and $\{b^i\},$ there exists a reversible transformation $T$ such that $T(a^i)=b^i$ for all $i$.
\end{principle}

An example in quantum theory is the Hadamard transformation reversibly mapping between the bases $\{\ket{0},\ket{1}\}$ and $\{\ket{+},\ket{-}\}$.

These last two principles imply that one can encode classical data in a system, and moreover, that any encoding is equivalent. In other words, information is independent of the encoding medium.

Principles $1$ to $4$ imply the following result (see \cite{chiribella2015operational} for a proof): for any given state $s$, there exists a natural number $n$ and a set of pure and perfectly distinguishable states $\{a^i\}_{i=1}^n$ such that $s=\sum_i p_ia_i$ where $0 \leq p_i \leq 1,\ \forall i$ and $\sum_i p_i=1$.

This result, together with principle $5$, implies the existence of a ``self-dualising'' \cite{muller2012structure,barnum2014higher} inner product $\langle {\cdot,\cdot}\rangle$. That is, to every pure state $s$, there is associated a unique pure effect $e^s$,
%\footnote{The proof of this fact requires a further technical assumption, implicit in \cite{muller2012structure}:  every mathematically allowed effect is physical. The first two axioms of the reconstruction of quantum theory provided in \cite{barnum2014higher} imply that the  every mathematically allowed effect is physical. As principles 1 to 5 from this paper imply the first two axioms of this reconstruction (see \cite{chiribella2015operational} for a proof of this), all mathematically allowed effects are physical in any theory satisfying our five principles.}
satisfying $e^s(s)=1$, such that: $e^s(\cdot)=\langle s, \cdot\rangle $. This inner product is invariant under all reversible transformations; satisfies $0 \leq \langle r, s \rangle \leq 1$ for all states $r, s$; $\langle s, s\rangle=1$ for all pure states $s$; and $\langle s, r \rangle=0$ if $s$ and $r$ are perfectly distinguishable. It also gives rise to the norm $\Vert \cdot \Vert =\sqrt{\langle{\cdot,\cdot}\rangle}$, satisfying $\Vert s \Vert \leq 1$ for all states $s$, with equality for pure states. We will make use of this norm in proving our main result.

%The above physical principles can be interpreted as the defining characteristics of information: information can be encoded in and extracted from physical mediums (Pure Sharp Effect); information is independent of  the encoding medium (Strong Symmetry); information propagates from present to future (Causality); and information is conservation at a fundamental level (Purification and Purity Preservation). It would therefore be surprising if these principles were not necessary primers for information processing and computation.
\section{Higher-order interference} 

Informally, a theory is said to have $n$th order interference if one can generate interference patterns in an $n$-slit experiment which cannot be created in any experiment with only $m$-slits, for all $m<n$.
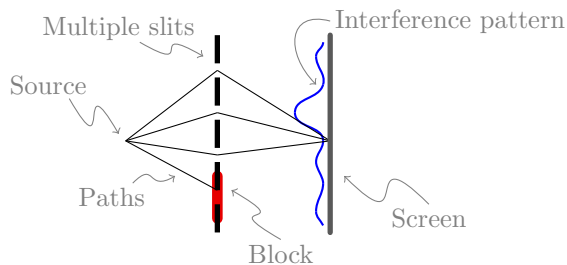
\[\begin{tikzpicture}[scale=0.75]
	\begin{pgfonlayer}{nodelayer}
		\node [style=none] (0) at (0, 1.75) {};
		\node [style=none] (1) at (0, 2.75) {};
		\node [style=none] (2) at (0, 1.25) {};
		\node [style=none] (3) at (0, 0.2500002) {};
		\node [style=none] (4) at (0, -1.75) {};
		\node [style=none] (5) at (0, -2.75) {};
		\node [style=none] (6) at (0, -0.2500002) {};
		\node [style=none] (7) at (0, -1.25) {};
		\node [style=none] (8) at (-3.25, -1) {};
		\node [style=none] (9) at (0, 1.5) {};
		\node [style=none] (10) at (0, -0) {};
		\node [style=none] (11) at (3.999999, -1) {};
		\node [style=none] (12) at (0, -1.5) {};
		\node [style=none] (13) at (0, -3.25) {};
		\node [style=none] (14) at (0, -4.25) {};
		\node [style=none] (15) at (0, -2.25) {};
		\node [style=none] (16) at (0, -3.75) {};
		\node [style=none] (17) at (0, -3.000001) {};
		\node [style=none] (18) at (3.999999, 2.75) {};
		\node [style=none] (19) at (3.999999, -4.25) {};
		\node [style=none] (20) at (-3.5, -0.7500001) {};
		\node [style=none] (21) at (-5.25, 0.5000001) {};
		\node [style=none] (22) at (-0.5000001, 2) {};
		\node [style=none] (23) at (-2.75, 2.5) {};
		\node [style=none] (24) at (-1.5, -2.25) {};
		\node [style=none] (25) at (-3, -2.5) {};
		\node [style=none] (26) at (0.4999992, -2.749999) {};
		\node [style=none] (27) at (2.25, -4.5) {};
		\node [style=none] (28) at (4.499999, -2.25) {};
		\node [style=none] (29) at (7.500001, -3.25) {};
		\node [style=none] (30) at (3.5, 1.000001) {};
		\node [style=none] (31) at (4, 3.25) {};
		\node [style=none, color=gray] (32) at (-3.5, 3) {Multiple slits};
		\node [style=none, color=gray] (33) at (-6, 1.000001) {Source};
		\node [style=none, color=gray] (34) at (-3.75, -3) {Paths};
		\node [style=none, color=gray] (35) at (2.25, -5) {Block};
		\node [style=none, color=gray] (36) at (7.5, -3.75) {Screen};
		\node [style=none, color=gray] (37) at (8.25, 3.25) {Interference pattern};
		\node [style=none] (38) at (6.75, 3) {};
		\node [style=none] (39) at (3.749999, 2.5) {};
		\node [style=none] (40) at (3.5, 2) {};
		\node [style=none] (41) at (3.749999, 1.000001) {};
		\node [style=none] (42) at (2.75, -0) {};
		\node [style=none] (43) at (3.75, -0.9999999) {};
		\node [style=none] (44) at (3.5, -1.749999) {};
		\node [style=none] (45) at (3.749999, -2.5) {};
		\node [style=none] (46) at (3.5, -3.5) {};
		\node [style=none] (47) at (3.75, -4) {};
	\end{pgfonlayer}
	\begin{pgfonlayer}{edgelayer}
		\draw [thick, color=blue, in=90, out=-90, looseness=1.00] (39.center) to (40.center);
		\draw [thick, color=blue, in=75, out=-90, looseness=1.00] (40.center) to (41.center);
		\draw [thick, color=blue, in=90, out=-105, looseness=1.00] (41.center) to (42.center);
		\draw [thick, color=blue, in=90, out=-90, looseness=1.00] (42.center) to (43.center);
		\draw [thick, color=blue, in=90, out=-90, looseness=1.00] (43.center) to (44.center);
		\draw [thick, color=blue, in=90, out=-90, looseness=0.75] (44.center) to (45.center);
		\draw [thick, color=blue, in=90, out=-90, looseness=1.00] (45.center) to (46.center);
		\draw [thick, color=blue, in=117, out=-90, looseness=1.00] (46.center) to (47.center);
		\draw [style=block] (15.center) to (16.center);
		\draw [style=slit] (1.center) to (0.center);
		\draw [style=slit] (2.center) to (3.center);
		\draw [style=slit] (6.center) to (7.center);
		\draw [style=slit] (4.center) to (5.center);
		\draw [style=none, color=black] (8.center) to (9.center);
		\draw [style=none, color=black] (8.center) to (10.center);
		\draw [style=none, color=black] (10.center) to (11.center);
		\draw [style=none, color=black] (9.center) to (11.center);
		\draw [style=none, color=black] (8.center) to (12.center);
		\draw [style=none, color=black] (12.center) to (11.center);
		\draw [style=slit] (13.center) to (14.center);
		\draw [style=none, color=black] (8.center) to (5.center);
		\draw [style=screen] (18.center) to (19.center);
		\draw [style=diredge, color=gray, in=120, out=-60, looseness=3.75] (21.center) to (20.center);
		\draw [style=diredge, color=gray, in=150, out=-30, looseness=2.75] (23.center) to (22.center);
		\draw [style=diredge, color=gray, in=-150, out=45, looseness=2.00] (25.center) to (24.center);
		\draw [style=diredge, color=gray, in=-30, out=150, looseness=3.00] (27.center) to (26.center);
		\draw [style=diredge, color=gray, in=-45, out=135, looseness=3.00] (29.center) to (28.center);
		\draw [style=diredge, color=gray, in=135, out=165, looseness=1.75] (31.center) to (30.center);
	\end{pgfonlayer}
\end{tikzpicture} \]
More precisely, this means that the interference pattern created on the screen cannot be written as a particular linear combination of the patterns generated when different subsets of slits are blocked. In the two slit experiment, quantum interference corresponds to the fact that the interference pattern cannot be written as the sum of the single slit patterns:
 
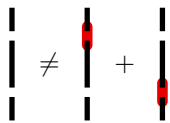
\[\begin{tikzpicture}[scale=0.50]
	\begin{pgfonlayer}{nodelayer}
		\node [style=none] (0) at (0, 1.25) {};
		\node [style=none] (1) at (0, 1.25) {};
		\node [style=none] (2) at (0, -1.25) {};
		\node [style=none] (3) at (0, -1.75) {};
		\node [style=none] (4) at (0, -3) {};
		\node [style=none] (5) at (0, 1.75) {};
		\node [style=none] (6) at (0, 3) {};
		\node [style=none] (7) at (2, -0) {$\neq$};
		\node [style=none] (8) at (4, 1.25) {};
		\node [style=none] (9) at (4, 1.25) {};
		\node [style=none] (10) at (4, -1.25) {};
		\node [style=none] (11) at (4, -1.75) {};
		\node [style=none] (12) at (4, -3) {};
		\node [style=none] (13) at (4, 1.75) {};
		\node [style=none] (14) at (4, 3) {};
		\node [style=none] (15) at (4, 2) {};
		\node [style=none] (16) at (4, 1) {};
		\node [style=none] (17) at (6, -0) {$+$};
		\node [style=none] (18) at (8, 1.25) {};
		\node [style=none] (19) at (8, 1.25) {};
		\node [style=none] (20) at (8, -1.25) {};
		\node [style=none] (21) at (8, -1.75) {};
		\node [style=none] (22) at (8, -3) {};
		\node [style=none] (23) at (8, 1.75) {};
		\node [style=none] (24) at (8, 3) {};
		\node [style=none] (25) at (8, -2) {};
		\node [style=none] (26) at (8, -0.9999998) {};
	\end{pgfonlayer}
	\begin{pgfonlayer}{edgelayer}
		\draw [style=block] (15.center) to (16.center);
		\draw [style=block] (25.center) to (26.center);
		\draw [style=slit] (1.center) to (2.center);
		\draw [style=slit] (6.center) to (5.center);
		\draw [style=slit] (3.center) to (4.center);
		\draw [style=slit] (9.center) to (10.center);
		\draw [style=slit] (14.center) to (13.center);
		\draw [style=slit] (11.center) to (12.center);
		\draw [style=slit] (19.center) to (20.center);
		\draw [style=slit] (24.center) to (23.center);
		\draw [style=slit] (21.center) to (22.center);
	\end{pgfonlayer}
\end{tikzpicture} \]
It was first shown by Sorkin \cite{sorkin1994quantum,sorkin1995quantum} that---at least for ideal experiments \cite{sinha2015superposition}---quantum theory is limited to the $n=2$ case. That is, the interference pattern created in a three---or more---slit experiment \emph{can} be written in terms of the two and one slit interference patterns obtained by blocking some of the slits. %; no genuinely new features result from considering three slits instead of two.
 Schematically:
%in a three slit experiment the interference pattern can be written in terms of the pairwise and single slit interference patterns, i.e.,
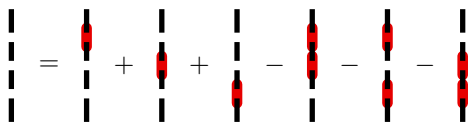
\[\begin{tikzpicture}[scale=0.50]
	\begin{pgfonlayer}{nodelayer}
		\node [style=none] (0) at (0, -0.25) {};
		\node [style=none] (1) at (0, 0.25) {};
		\node [style=none] (2) at (0, 1.25) {};
		\node [style=none] (3) at (0, -1.25) {};
		\node [style=none] (4) at (0, -1.75) {};
		\node [style=none] (5) at (0, 1.75) {};
		\node [style=none] (6) at (0, 3) {};
		\node [style=none] (7) at (0, -3) {};
		\node [style=none] (8) at (4, -1.25) {};
		\node [style=none] (9) at (4, -1.75) {};
		\node [style=none] (10) at (4, -0.25) {};
		\node [style=none] (11) at (4, -3) {};
		\node [style=none] (12) at (4, 1.75) {};
		\node [style=none] (13) at (4, 3) {};
		\node [style=none] (14) at (4, 1.25) {};
		\node [style=none] (15) at (4, 0.25) {};
		\node [style=none] (16) at (12, -3) {};
		\node [style=none] (17) at (8, 1.75) {};
		\node [style=none] (18) at (8, -3) {};
		\node [style=none] (19) at (12, 1.25) {};
		\node [style=none] (20) at (12, -0.25) {};
		\node [style=none] (21) at (8, 3) {};
		\node [style=none] (22) at (8, -0.25) {};
		\node [style=none] (23) at (12, -1.75) {};
		\node [style=none] (24) at (8, 1.25) {};
		\node [style=none] (25) at (12, 3) {};
		\node [style=none] (26) at (8, 0.25) {};
		\node [style=none] (27) at (12, 0.25) {};
		\node [style=none] (28) at (8, -1.75) {};
		\node [style=none] (29) at (8, -1.25) {};
		\node [style=none] (30) at (12, -1.25) {};
		\node [style=none] (31) at (12, 1.75) {};
		\node [style=none] (32) at (20, -1.75) {};
		\node [style=none] (33) at (24, 3) {};
		\node [style=none] (34) at (24, -1.25) {};
		\node [style=none] (35) at (16, -3) {};
		\node [style=none] (36) at (20, -0.2500001) {};
		\node [style=none] (37) at (20, -1.25) {};
		\node [style=none] (38) at (20, 3) {};
		\node [style=none] (39) at (24, 0.2500001) {};
		\node [style=none] (40) at (16, -1.75) {};
		\node [style=none] (41) at (16, -1.25) {};
		\node [style=none] (42) at (24, -0.2500001) {};
		\node [style=none] (43) at (24, 1.75) {};
		\node [style=none] (44) at (16, -0.2500001) {};
		\node [style=none] (45) at (16, 0.2500001) {};
		\node [style=none] (46) at (20, 1.75) {};
		\node [style=none] (47) at (20, 0.2500001) {};
		\node [style=none] (48) at (20, 1.25) {};
		\node [style=none] (49) at (20, -3) {};
		\node [style=none] (50) at (24, -1.75) {};
		\node [style=none] (51) at (24, -3) {};
		\node [style=none] (52) at (24, 1.25) {};
		\node [style=none] (53) at (16, 1.25) {};
		\node [style=none] (54) at (16, 3) {};
		\node [style=none] (55) at (16, 1.75) {};
		\node [style=none] (56) at (2, -0) {$=$};
		\node [style=none] (57) at (6, -0) {$+$};
		\node [style=none] (58) at (10, -0) { $+$};
		\node [style=none] (59) at (14, -0) { $-$};
		\node [style=none] (60) at (18, -0) { $-$};
		\node [style=none] (61) at (22, -0) { $-$};
		\node [style=none] (62) at (4, 1) {};
		\node [style=none] (63) at (4, 2) {};
		\node [style=none] (64) at (8, 0.4999999) {};
		\node [style=none] (65) at (8, -0.4999999) {};
		\node [style=none] (66) at (12, -0.9999998) {};
		\node [style=none] (67) at (12, -2) {};
		\node [style=none] (68) at (16, 2) {};
		\node [style=none] (69) at (16, 1) {};
		\node [style=none] (70) at (20, 2) {};
		\node [style=none] (71) at (20, 1) {};
		\node [style=none] (72) at (20, -0.9999998) {};
		\node [style=none] (73) at (20, -2) {};
		\node [style=none] (74) at (24, -2) {};
		\node [style=none] (75) at (24, -0.7500001) {};
		\node [style=none] (76) at (24, -0.9999998) {};
		\node [style=none] (77) at (16, -0.4999999) {};
		\node [style=none] (78) at (16, 0.4999999) {};
		\node [style=none] (79) at (24, -0.4999999) {};
		\node [style=none] (80) at (24, 0.4999999) {};
	\end{pgfonlayer}
	\begin{pgfonlayer}{edgelayer}
		\draw [style=block] (62.center) to (63.center);
		\draw [style=block] (65.center) to (64.center);
		\draw [style=block] (67.center) to (66.center);
		\draw [style=block] (69.center) to (68.center);
		\draw [style=block] (71.center) to (70.center);
		\draw [style=block] (73.center) to (72.center);
		\draw [style=block] (74.center) to (76.center);
		\draw [style=block] (77.center) to (78.center);
		\draw [style=block] (79.center) to (80.center);
		\draw [style=slit] (5.center) to (6.center);
		\draw [style=slit] (1.center) to (2.center);
		\draw [style=slit] (3.center) to (0.center);
		\draw [style=slit] (7.center) to (4.center);
		\draw [style=slit] (12.center) to (13.center);
		\draw [style=slit] (15.center) to (14.center);
		\draw [style=slit] (8.center) to (10.center);
		\draw [style=slit] (11.center) to (9.center);
		\draw [style=slit] (17.center) to (21.center);
		\draw [style=slit] (26.center) to (24.center);
		\draw [style=slit] (29.center) to (22.center);
		\draw [style=slit] (18.center) to (28.center);
		\draw [style=slit] (31.center) to (25.center);
		\draw [style=slit] (27.center) to (19.center);
		\draw [style=slit] (30.center) to (20.center);
		\draw [style=slit] (16.center) to (23.center);
		\draw [style=slit] (55.center) to (54.center);
		\draw [style=slit] (45.center) to (53.center);
		\draw [style=slit] (41.center) to (44.center);
		\draw [style=slit] (35.center) to (40.center);
		\draw [style=slit] (46.center) to (38.center);
		\draw [style=slit] (47.center) to (48.center);
		\draw [style=slit] (37.center) to (36.center);
		\draw [style=slit] (49.center) to (32.center);
		\draw [style=slit] (43.center) to (33.center);
		\draw [style=slit] (39.center) to (52.center);
		\draw [style=slit] (34.center) to (42.center);
		\draw [style=slit] (51.center) to (50.center);
	\end{pgfonlayer}
\end{tikzpicture} \]
If a theory does not have $n$th order interference then one can show it will not have $m$th order interference, for any $m>n$ \citep{sorkin1994quantum}. As such, one can classify theories according to their maximal order of interference, $\OInt$. For example quantum theory lies at $\OInt=2$ and classical theory at $\OInt=1$.

Higher order interference was initially formalised by Sorkin in the framework of Quantum Measure Theory \cite{sorkin1994quantum} but has more recently been adapted to the setting of generalised probabilistic theories in \cite{barnum2014higher,lee2015generalised,ududec2011three,lee2015higher}. The most direct translation to this setting describes the order of interference in terms of probability distributions corresponding to the different experimental setups (which slits are open, etc.) \cite{lee2015generalised}. However, given our five principles, it is possible to define physical transformations that correspond to the action of blocking certain subsets of slits. In this case, there is a more convenient (and equivalent, given the five princples) definition in terms of such transformations \cite{barnum2014higher}.

If there are $N$ slits, labelled $1, \dots, N$, these transformations are denoted $P_I$, where $I \subseteq \{1, \dots, N\}:=\mathbf{N}$ corresponds to the subset of slits which are not blocked. In general we expect that $P_I P_J = P_{I\cap J}$, as only those slits belonging to both $I$ and $J$ will not be blocked by either $P_I$ or $P_J$. This intuition suggests that these transformations should correspond to projectors (i.e. idempotent transformations $P_IP_I=P_I$). Given principles $1$ to $5$, it was shown in \cite{barnum2014higher} that this is indeed the case. Given this structure, one can define the maximal order of interference as follows \cite{barnum2014higher}.

\begin{definition}
A theory satisfying principles $1$ to $5$ has maximal order of interference $\OInt$ if, for any $N \geq \OInt$, one has:
\[\mathds{1}_N = \sum_{{\small \begin{array}{c} I\subseteq \mathbf{N} \\ |I|\leq \OInt \end{array}}}\mathcal{C}\left(\OInt,|I|,N\right)P_I\]
where $\mathds{1}_N $ is the identity on a system with $N$ pure and perfectly distinguishable states and
\[\mathcal{C}\left(\OInt,|I|,N\right):=(-1)^{\OInt-|I|}\left(\begin{array}{c}N-|I|-1\\\OInt-|I|\end{array}\right)\]
\end{definition}

%these projectors then for a $\OInt$ order theory there is a decomposition of the identity on an $N$ level system in terms of these projectors. The existence of such a decomposition can then be used to define the order of interference of the theory.
The factor $\mathcal{C}\left(\OInt,|I|,N\right)$ in the above definition corrects for the overlaps that occur when different combinations of slits are blocked. %The above definition says that for an $N>\OInt$ level system, the identity -- which corresponds to having all slits open -- can be decomposed as a linear combination of projectors onto the subsets of slits of size equal or smaller than $\OInt$.
Note that, for the case $\OInt=N$, this reduces to the expected expression of $\mathds{1}_h=P_{\{1,...,\OInt\}}$ i.e. the identity is given by the projector with all slits open. The case of $N=\OInt+1$ corresponds to $\mathcal{C}\left(\OInt,|I|,\OInt+1\right) = (-1)^{\OInt-|I|}$, which is the situation depicted in the previous figures, as well as the one most commonly discussed in the literature \cite{sorkin1994quantum,ududec2011three}.

Rather than work directly with these physical projectors, it is mathematically more convenient to work with (generally) unphysical transformations corresponding to projectors onto the ``coherences'' of a state. For example, in the case of a qutrit, the projector $P_{\{0,1\}}$ projects onto a two dimensional subspace:
\[ P_{\{0,1\}}::\left(\begin{array}{ccc} \rho_{00} & \rho_{01} & \rho_{02} \\ \rho_{10} & \rho_{11} & \rho_{12} \\\rho_{20} & \rho_{21} & \rho_{22}   \end{array}\right)\mapsto \left(\begin{array}{ccc} \rho_{00} & \rho_{01} & 0 \\ \rho_{10} & \rho_{11} & 0 \\0 & 0 & 0   \end{array}\right)\] 
whilst the coherence-projector $\omega_{\{0,1\}}$ projects only onto the coherences in that two dimensional subspace:
\[ \omega_{\{0,1\}}::\left(\begin{array}{ccc} \rho_{00} & \rho_{01} & \rho_{02} \\ \rho_{10} & \rho_{11} & \rho_{12} \\\rho_{20} & \rho_{21} & \rho_{22}   \end{array}\right)\mapsto \left(\begin{array}{ccc} 0 & \rho_{01} & 0 \\ \rho_{10} & 0 & 0 \\0 & 0 & 0   \end{array}\right).\]
That is, $\omega_{\{0,1\}}$ corresponds to the linear combination of projectors: $P_{\{0,1\}}-P_{\{0\}}-P_{\{1\}}$.

There is a coherence-projector $\omega_I$ for each subset of slits $I \subseteq \mathbf{N}$, defined in terms of the physical projectors:
$$\omega_I:=\sum_{\tilde{I}\subseteq I}(-1)^{|I|+|\tilde{I}|}P_{\tilde{I}}.$$ These have the following useful properties, proved in appendix~\ref{AppendixCoherence}.
\begin{lemma} \label{lemma: decompisition of the identity into coherences}
An equivalent definition of the maximal order of interference, $\OInt$, is: $\mathds{1}_N=\sum_{I,|I|=1}^\OInt\omega_I,$ for all $ N \geq \OInt.$
\end{lemma}
The above lemma implies that any state (indeed, any vector in the vector space generated by the states) can be decomposed as $s=\sum_{I,|I|=1}^\OInt s_I,$ where $s_I:=\omega_Is$.
\begin{lemma} \label{lemma: orthogonality of coherences}
``Coherences are orthogonal'': $\mathrm{i)}$ $\omega_I\omega_J=\delta_{IJ}\omega_I$, for all $I,J$ and $\mathrm{ii)}$ $\Norm{s}^2=\sum_I\lVert\omega_Is\rVert^2$
\end{lemma}

\section{Setting up the problem} 
In the standard search problem, one is asked to find a specific ``marked'' item from among a large collection of items in some unstructured list. The items are indexed $1, \dots, N$ and one has access to an oracle, which, when asked whether item $i$ is the marked item, denoted $x$, returns the answer ``yes'' or ``no''. Informally, the search problem asks for the minimal number of queries to this oracle required to find $x$ in the worst case.
%In the standard bra-ket formalism of quantum theory, this oracle corresponds to a controlled unitary transformation $U$ defined by its action on the (product) computational basis:
%$ U |y\rangle |q\rangle = |x\rangle |q\oplus f(y)\rangle,$
%where $|y\rangle$ is the index or control register, $|q\rangle$ is the target register, $\oplus$ denotes addition modulo 2 and $f:\{0,\dots, N-1\}\rightarrow\{0,1\}$ satisfies $f(y)=1$ if $x$ is the marked item and $f(y)=0$ otherwise. Inputting $|-\rangle$ into the target register results in a phase being ``kicked-back'' to the control register: $U|y \rangle |-\rangle=(-1)^{f(y)}|y \rangle |-\rangle.$ Disregarding the target register reduces the action of the oracle to applying the phase transformation $O_x|y \rangle =(-1)^{f(y)}|y \rangle$, where $x$ here denotes the marked item $f(y)=1$ if and only if $y=x$.

In the standard bra-ket formalism of quantum theory, this oracle corresponds to a controlled unitary transformation $U$, defined by its action on the (product) computational basis:
$ U |i\rangle |q\rangle = |i\rangle |q\oplus f(i)\rangle,$
where $|i\rangle$ is the index, or control, register, $|q\rangle$ is the target register, $\oplus$ denotes addition modulo 2 and $f:\{1,\dots, N\}\rightarrow\{0,1\}$ satisfies $f(i)=1$ if and only if $i=x$. Inputting $|-\rangle$ into the target register results in a phase being ``kicked-back'' to the control register: $U|i \rangle |-\rangle=(-1)^{f(i)}|i\rangle|-\rangle.$ Discarding the target register reduces the action of the oracle to applying the phase transformation $O_{x}|i \rangle =(-1)^{f(i)}|i\rangle$.
%Changing to the density matrix formalism, we see that this phase oracle acts as the identity on all quantum coherences $\omega_I$ for which the marked item $x$ is not an element of the set $I$, and changes all $\omega_J$ with $x\in J$ to $-\omega_J$.
Changing to the density matrix formalism, we see that this phase oracle, whose action on states $\rho$ is now denoted by $\Ora{x}\rho$, %\footnote{$\Ora{x}$ is related to $O_x$ via $\Ora{x}\rho=O_x\rho O_x^T$}
acts as the identity on the diagonal elements of all density matrices whilst adding a `$-$' to the off diagonal elements $\{\rho_{xi},\rho_{ix}\}_{i}.$ %where $x$ is the marked item and $j$ any other item.

Previous work \cite{lee2015generalised} has shown that the conjunction of principles 1, 2, 3 and 5 implies the existence of reversible controlled transformations. These can be used to define oracles in a manner analogous to quantum theory \cite{lee2015generalised}. %For a given reversible controlled transformation in the theory, $C$, there exists a set of pure and perfectly distinguishable states $\{s^i\}_{i=1}^N$ and a collection of reversible transformations $\{T_i\}_{i=1}^N$, such that $C\left(s^i \otimes w\right) = s^i \otimes T_i(w)$, for all states $w$.
Moreover, every controlled transformation gives rise to a ``kicked-back'' reversible phase transformation on the control system \cite{lee2015generalised}. Thus---as in quantum theory---from the point of view of querying the oracle, we can reduce all considerations involving the controlled transformation to those involving the kicked-back phase.

To highlight the role of interference in searching an unstructured list, we describe the action of querying the oracle in terms of the physically motivated set-up of $N$-slit experiments. %we relate the problem to the physical set-up of $N$-slit experiments, presented in section~\ref{Section: Higher-order interference}.
Consider first the quantum case. Note that an $N$-slit experiment defines a set of $N$ pure and perfectly distinguishable states $\ket{i}\bra{i}$, each of which can be associated to a distinct element in the $N$ item list. Querying the oracle about item $i$ is equivalent to applying the oracle transformation to state $\ket{i}\bra{i}$. In quantum theory, preparing such a state can be achieved by passing a uniform superposition through the $N$-slit experiment with all but the $i$th slit blocked. The oracle can be implemented by placing a phase shifter behind slit $x$. Querying the oracle in a superposition of states can then be achieved by varying which slits are blocked. This is illustrated schematically below:
%We associate each element in the list to a single slit in an $N$-slit experiment. Given this set-up, querying the oracle in state $s^i$ can be implemented by first passing some `superposition' state through the $N$-slit experiment when all but the $i$th slit is blocked before the oracle is queried. Moreover, we can query the oracle in an arbitrary state by Second, rather than directly preparing the state we wish to query, we pass a fixed state through the $N$-slit experiment before it queries the oracle, where the choice of blocked slits transform the fixed input state to the desired query state.
%This set-up connects the action of querying the oracle to the process of opening and closing different subsets of slits, depicted schematically below.
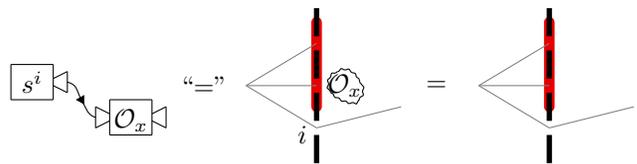
\[\begin{tikzpicture}[scale=0.75]
	\begin{pgfonlayer}{nodelayer}
		\node [style=none] (0) at (1, 1.25) {};
		\node [style=none] (1) at (1.5, 1) {};
		\node [style=none] (2) at (1.5, 0.75) {};
		\node [style=none] (3) at (1, 0.5) {};
		\node [style=none] (4) at (1.5, 1.5) {};
		\node [style=none] (5) at (3, 1.5) {};
		\node [style=none] (6) at (3, 0.25) {};
		\node [style=none] (7) at (1.5, 0.25) {};
		\node [style=none] (8) at (3.5, 0.5) {};
		\node [style=none] (9) at (3, 0.75) {};
		\node [style=none] (10) at (3.5, 1.25) {};
		\node [style=none] (11) at (3, 1) {};
		\node [style=none] (12) at (-2, 2.75) {};
		\node [style=none] (13) at (-0.5, 1.5) {};
		\node [style=none] (14) at (-2, 1.5) {};
		\node [style=none] (15) at (-0.5, 2.75) {};
		\node [style=none] (16) at (-0.5, 2) {};
		\node [style=none] (17) at (-0.5, 2.25) {};
		\node [style=none] (18) at (0, 2.5) {};
		\node [style=none] (19) at (0, 1.75) {};
		\node [style=none] (20) at (0, 2.125) {};
		\node [style=arrowhead, rotate=30] (21) at (0.5, 1.5) {};
		\node [style=none] (22) at (3.5, 0.75) {};
		\node [style=none] (23) at (-1.25, 2.125) {$s^i$};
		\node [style=none] (24) at (2.25, 0.75) {$\mathcal{O}_x$};
		\node [style=none] (25) at (1, 0.875) {};
	\end{pgfonlayer}
	\begin{pgfonlayer}{edgelayer}
		\draw (0.center) to (1.center);
		\draw (3.center) to (2.center);
		\draw (4.center) to (7.center);
		\draw (7.center) to (6.center);
		\draw (6.center) to (5.center);
		\draw (5.center) to (4.center);
		\draw (10.center) to (11.center);
		\draw (8.center) to (9.center);
		\draw [style=none] (0.center) to (3.center);
		\draw [style=none] (10.center) to (8.center);
		\draw (12.center) to (14.center);
		\draw (14.center) to (13.center);
		\draw (13.center) to (15.center);
		\draw (15.center) to (12.center);
		\draw (18.center) to (17.center);
		\draw (19.center) to (16.center);
		\draw [style=none] (18.center) to (19.center);
		\draw [style=none, color=black, in=120, out=-15, looseness=1.25] (20.center) to (21.center);
		\draw [style=none, color=black, in=165, out=-60, looseness=1.00] (21.center) to (25.center);
	\end{pgfonlayer}
\end{tikzpicture}
\begin{tikzpicture}[scale=0.75]
	\begin{pgfonlayer}{nodelayer}
		\node [style=none] (0) at (1.25, 3.75) {};
		\node [style=none] (1) at (1.25, 4.75) {};
		\node [style=none] (2) at (2, 2.5) {};
		\node [style=none] (3) at (2.75, 2.25) {};
		\node [style=none] (4) at (2.75, 1.5) {};
		\node [style=none] (5) at (1.75, 1.75) {};
		\node [style=none] (6) at (2.25, 2) {$\Ora{x}$};
		\node [style=none] (7) at (1.25, 3.25) {};
		\node [style=none] (8) at (1.25, 2.25) {};
		\node [style=none] (9) at (1.25, 0.2499996) {};
		\node [style=none] (10) at (1.25, -0.7499997) {};
		\node [style=none] (11) at (1.25, 1.75) {};
		\node [style=none] (12) at (1.25, 0.7500002) {};
		\node [style=none] (13) at (1.75, 2.75) {};
		\node [style=none] (14) at (-1.25, 2) {};
		\node [style=none] (15) at (1.25, 0.5000001) {};
		\node [style=none] (16) at (4.25, 1.25) {};
		\node [style=none] (17) at (1.25, 4.25) {};
		\node [style=none] (18) at (1.25, 1.25) {};
		\node [style=none] (19) at (5.5, 2) {$=$};
		\node [style=none] (20) at (9.499999, 3.75) {};
		\node [style=none] (21) at (9.499999, 0.2499996) {};
		\node [style=none] (22) at (9.499999, 2.25) {};
		\node [style=none] (23) at (9.499999, 0.5000001) {};
		\node [style=none] (24) at (9.499999, 1.25) {};
		\node [style=none] (25) at (9.499999, -0.7499997) {};
		\node [style=none] (26) at (12.5, 1.25) {};
		\node [style=none] (27) at (9.499999, 1.75) {};
		\node [style=none] (28) at (9.499999, 4.25) {};
		\node [style=none] (29) at (9.499999, 4.75) {};
		\node [style=none] (30) at (9.499999, 3.25) {};
		\node [style=none] (31) at (9.499999, 0.7500002) {};
		\node [style=none] (32) at (7, 2) {};
		\node [style=none] (33) at (1.25, 3.5) {};
		\node [style=none] (34) at (1.25, 2) {};
		\node [style=none] (35) at (9.499999, 3.5) {};
		\node [style=none] (36) at (9.499999, 2) {};
		\node [style=none] (37) at (11.5, 1) {};
		\node [style=none] (38) at (0.7499997, 0.2499996) {$i$};
		\node [style=none] (39) at (-2.75, 2) {``$=$''};
		\node [style=none] (40) at (-4, 2) {};
	\end{pgfonlayer}
	\begin{pgfonlayer}{edgelayer}
		\draw [style=block] (28.center) to (24.center);
		\draw [style=block] (17.center) to (18.center);
		\draw [style=slit] (1.center) to (0.center);
		\draw [style=wavy, in=105, out=-165, looseness=1.00] (2.center) to (5.center);
		\draw [style=wavy, in=-135, out=-45, looseness=0.75] (5.center) to (4.center);
		\draw [style=wavy, in=-75, out=45, looseness=1.00] (4.center) to (3.center);
		\draw [style=wavy, bend right, looseness=1.25] (3.center) to (2.center);
		\draw [style=slit] (7.center) to (8.center);
		\draw [style=slit] (11.center) to (12.center);
		\draw [style=slit] (9.center) to (10.center);
		\draw [style=none, color=gray] (14.center) to (15.center);
		\draw [style=none, color=gray] (15.center) to (16.center);
		\draw [style=slit] (29.center) to (20.center);
		\draw [style=slit] (30.center) to (22.center);
		\draw [style=slit] (27.center) to (31.center);
		\draw [style=slit] (21.center) to (25.center);
		\draw [style=none, color=gray] (32.center) to (23.center);
		\draw [style=none, color=gray] (23.center) to (26.center);
		\draw [style=none, color=gray] (14.center) to (33.center);
		\draw [style=none, color=gray] (14.center) to (34.center);
		\draw [style=none, color=gray] (32.center) to (35.center);
		\draw [style=none, color=gray] (32.center) to (36.center);
	\end{pgfonlayer}
\end{tikzpicture} \]
As discussed previously, the physical act of blocking slits is represented by the projectors $P_I$. The action of the quantum oracle can thus be rephrased in terms of these projectors: i) $\Ora{x} P_I=P_I$, if $x \notin I$ or $|I|=1$ and, ii) $\Ora{x}$ can act non-trivially on projectors $P_I$ with $x \in I$ and $|I|>1$, but must satisfy $\Ora{x}P_I=P_I\Ora{x}$, for all $P_I$, which corresponds to the fact that a quantum oracle does not ``create'' or ``destroy'' coherence between states passing through different slits.

By analogy with the quantum case we can define the oracle which encodes the search problem in theories satisfying principles $1$ to $5$ as follows. Note that in this paper we only deal with the case of a single marked item.
\begin{definition} \label{Definition: Search Oracle}
A reversible transformation is a \emph{search oracle}, denoted $\Ora{x}$, if and only if:
$$ \begin{aligned} \mathrm{i)} \ \Ora{x}P_I&=P_I  \text{ for} \ \text{all} \ x\notin I \ \text{or} \ |I|=1 \ \text{and,} \\
\mathrm{ii)} \ \Ora{x}P_I&=P_I \Ora{x}\text{, for all }P_I. \end{aligned}$$
\end{definition} 
In the above definition, the requirement $\Ora{x}P_I=P_I \Ora{x}$, for all $P_I$, is quite natural. This requirement ensures that one cannot gain any information about item $i$ when querying the oracle using a state with no support on $i$, i.e. a state $s$ such that $P_I s=s$ where $i \notin I$.
In an arbitrary theory, it may not be the case that a transformation satisfying definition~\ref{Definition: Search Oracle} and acting non-trivially on $P_I$, with $x \in I$, exists. This is not an issue as in such theories we cannot even define the search problem, let alone show it can be solved using fewer queries than quantum theory. In this work, we shall assume the existence of a search oracle in any theory we consider.
Given the definition of coherence-projectors $\omega_I$ we can equivalently write definition~\ref{Definition: Search Oracle} as: %are linear combinations of the $P_I$'s, one can equivalently write the action of a search oracle in terms of the $\omega_I$'s. Indeed, $\Ora{x}P_I=P_I$, for $x \notin I$ or $|I|=1$, implies 
$\Ora{x} \omega_I=\omega_I$, for $x \notin I$ or $|I|=1$, and
%$\Ora{x}P_I=P_I \Ora{x}$, for all $P_I$, implies
$\Ora{x}\omega_I=\omega_I \Ora{x}$, for all $I$.
Indeed, in the quantum case, the action of the oracle can be equivalently described as:
$\Ora{x} \omega_I =  \omega_I \ \text{ if } x\not\in I \ \text{ or } |I|=1 $, and $\Ora{x}\omega_I=  -\omega_I \ \text{ otherwise}$.

We can now formally state the search problem for a single marked item---defined for the quantum case in \cite{nielsen2010quantum,zalka1999grover,boyer1996tight}---as:

\begin{search problem}
Given an $N$ element list with search oracle $\Ora{x}$ and an arbitrary collection of reversible transformations $\{G_i\}$, what is the minimal $k \in \mathbb{N}$ such that $G_k\Ora{x}G_{k-1}\dots G_1\Ora{x}s$ can be found, with probability greater than $1/2$, to be in the state $x$, for arbitrary state $s$, averaged over all possible marked items?
\end{search problem}

\section{Main result}%We now prove our main result.
\begin{theorem} \label{Main Theorem}
In theories satisfying principles $1$ to $5$, with finite maximal order of interference $\OInt$, the number of queries needed to solve the search problem is $\Omega(\sqrt{N/\OInt})$.
\end{theorem}

\begin{proof}[Proof of theorem~\ref{Main Theorem}]
The basic idea is based on the proof of the quantum case presented in \cite{nielsen2010quantum,boyer1996tight,zalka1999grover}.
Let
$$\begin{aligned}
s_k^x&=G_k\Ora{x}G_{k-1}\dots G_1\Ora{x}s, \\
s_k&=G_kG_{k-1}\dots G_1s,
\end{aligned}$$
where $G_i$ is some reversible transformation from the theory, and define
$$D_k=\sum_x \Vert s_k^x - s_k \Vert^2.$$
It will be shown that, for $\langle x,s_k^x \rangle \geq 1/2$, we have $cN \leq D_k \leq 4hk^2$, where $c$ is any constant less than $\left(\sqrt{2}-1\right)^2$, from which the result $k \geq O\left(\sqrt{\frac{N}{h}}\right)$ follows. The lower bound goes through as in the quantum case and is derived in appendix~\ref{Appendix: lower bound for main proof}. The upper bound will now be proved by induction.

We have 
$$\begin{aligned}
D_{k+1}&=\sum_x \Vert G_{k+1} \left( \Ora{x} s_k^x -s_k\right)\Vert^2=\sum_x \Vert \Ora{x} s_k^x -s_k\Vert^2 \\
&=\sum_x \Vert \Ora{x}\left( s_k^x -s_k\right) + \left(\Ora{x}-\mathds{1}\right)s_k\Vert^2\\
&\leq \sum_x \Vert s_k^x-s_k\Vert^2 \\ 
& \quad \quad +2\sum_x \Vert \Ora{x}\left( s_k^x -s_k\right)\Vert \Vert \left(\Ora{x}-\mathds{1}\right)s_k\Vert \\
& \quad \quad \quad \quad \quad +\sum_x \Vert\left(\Ora{x}-\mathds{1}\right)s_k\Vert^2 \\
&\leq D_k +2\sqrt{D_k \sum_x \Vert\left(\Ora{x}-\mathds{1}\right)s_k\Vert^2} +\Vert\left(\Ora{x}-\mathds{1}\right)s_k\Vert^2 \\
&\leq \left(\sqrt{D_k} + \sqrt{\sum_x \Vert (\mathds{1}-\mathcal{O}_x)s_k \Vert^2} \right)^2,
\end{aligned}
$$
%$$D_{k+1} \leq \left(\sqrt{D_k} + \sqrt{\sum_x \Vert (\mathds{1}-\mathcal{O}_x)s_k \Vert^2} \right)^2,$$ 
which follows from the triangle inequality, the Cauchy-Schwarz inequality, and the fact the norm is invariant under reversible transformations. %The last equality followed from lemma~\ref{lemma: single query bound}.

The quantity $\sum_x\lVert(\mathds{1}-\mathcal{O}_x)s_k\rVert^2$---which can be thought of as how much some state is ``moved'' in a single query, averaged over all possible marked items $x$---is the only theory dependent quantity that features in this proof. We upper bound it as follows: %This is a measure of how effective $\Ora{x}$ is, for each $x$. This will, in general, depend on the specific theory that we are considering, so, rather than calculating this exactly, we prove an upper bound on it, resulting from principles $1$ to $5$.
%The sum over $x$ naively suggests that this quantity could grow with increasing $N$. However, all is not lost. Recall the decomposition of the state from lemma \ref{lemma: decompisition of the identity into coherences}, $s= \sum_I \omega_Is$, and the fact that the oracle can only effect the part corresponding to $x\in I$, with $|I|>1$. A simple counting argument implies that the ratio of the number of summands effected to the total number, goes as $O(\frac{1}{N})$. Based on this one might expect to see similar behaviour for $\lVert(\mathds{1}-\mathcal{O}_x)s\rVert^2$. This is indeed the case, as we shall now prove.
$$ \begin{aligned}  
&\sum_x \lVert(\mathds{1}-\mathcal{O}_x){s_k}\rVert^2  \\
&=\sum_x\sum_I \lVert(\mathds{1}-\mathcal{O}_x)\omega_I{s_k}\rVert^2\\
&=\sum_x\sum_{{\tiny \begin{array}{c}I \\|I|>1\\ x\in I\end{array}}} \lVert\omega_I(\mathds{1}-\mathcal{O}_x){s_k}\rVert^2\\
&\leq\sum_x\sum_{{\tiny \begin{array}{c}I \\|I|>1\\ x\in I\end{array}}} \left(\lVert\mathds{1}\omega_I{s_k}\rVert+\lVert\mathcal{O}_x\omega_I{s_k}\rVert\right)^2 \\
& \leq\sum_x\sum_{{\tiny \begin{array}{c}I \\|I|>1\\ x\in I\end{array}}} 4\lVert\omega_I{s_k}\rVert^2,
\end{aligned}$$
%$$\begin{aligned} &\leq\sum_x\sum_{{\tiny \begin{array}{c}I \\|I|>1\\ x\in I\end{array}}} (2\lVert\omega_I{s}\rVert)^2=\sum_x\sum_{{\tiny \begin{array}{c}I \\|I|>1\\ x\in I\end{array}}} 4\lVert\omega_I{s}\rVert^2,\end{aligned} $$ 
where the first line follows from lemma~\ref{lemma: decompisition of the identity into coherences}, lemma~\ref{lemma: orthogonality of coherences}, and the definition of the search oracle $\Ora{x}$, and second from the triangle inequality and the fact that the norm is invariant under reversible transformations. We need to know how many times each $\lVert\omega_I{s_k}\rVert^2$ appears when we sum over the marked item $x$. Each given $I=\{i_1, i_2, \dots, i_{|I|}\}$ will appear $|I|$ times as we sum over $x$, one for every time $i_j$ is the marked item. Thus

$$ \begin{aligned} 
&\sum_x \lVert(\mathds{1}-\mathcal{O}_x){s_k}\rVert^2 \leq \sum_{{\tiny \begin{array}{c}I \\|I|>1\end{array}}} 4|I|\lVert\omega_I{s_k}\rVert^2\\
&\leq 4 \sum_{I} |I|\lVert\omega_I{s_k}\rVert^2\leq 4\OInt \sum_{I} \lVert\omega_I{s_k}\rVert^2= 4\OInt \lVert{s_k}\rVert^2\leq 4\OInt.
\end{aligned}$$
The second line follows from $\sum_{|I|=1}\Vert \omega_I s_k \Vert^2 \geq 0$, lemma~\ref{lemma: orthogonality of coherences}, $\Vert s_k\Vert \leq 1$, and $|I| \leq \OInt$, for all $I$.
We thus have:
$
\begin{aligned}
D_{k+1}\leq \left(\sqrt{D_k} + \sqrt{4h}\right)^2. 
\end{aligned}$
Assuming that $D_k \leq 4hk^2$ gives us $D_{k+1} \leq 4h(k+1)^2,$ from which the result follows via induction.
\end{proof}
%Theorem~\ref{Main Theorem} immediately implies the following. \begin{corollary} In any theory satisfying principles $1$ to $5$ with any finite order of interference, there exists an oracle relative to which \textbf{NP}-complete problems cannot be solved efficiently. \end{corollary}

\section{Discussion} 
%This gives us a natural oracle relative to which these theories cannot solve NP complete problems.
%\begin{itemize}
%\item Previous work considering Grovers bound in post-quantum theories have used ad-hoc modifications of quantum theory which do not result in consistent theories. In contrast the GPT framework allow investigations into alternate theories in a physically and mathematically consistent manor.
%\item  This is the first proven consequence of having higher order interference and shows little advantage over qauntum theory, which is surprising given previous work proving that limiting interference gives quantum-like features, could quantum theory be the ``simplest'' theory that is able to achieve this speed-up? \end{itemize}
In this work, we considered theories satisfying certain natural physical principles which are sufficient for the existence of controlled transformations and a phase kick-back mechanism, necessary features for a well-defined search oracle. Given these physical principles, we proved that a theory with maximal order of interference $\OInt$ requires $\Omega(\sqrt{N/\OInt})$ queries to this oracle to find a single marked item from some $N$-element list. This result challenges our pre-conceived notions about how quantum computers achieve their computational advantage and is somewhat surprising as one might expect more interference to imply more computational power. %This suggests the hypothesis that a quantum computer can efficiently simulate computation in any theory with a well-defined notion of information. %For other papers investigating computation in the generalised probabilistic theories framework, see \cite{lee2015computation,lee2015proofs,barrett2015landscape}. %Could quantum theory be the ``simplest'' theory that is able to achieve the speed-up offered by Grover's algorithm? In any case, our result constitutes the first consequence of post-quantum interference.
Further work will focus on determining sufficient physical principles for there to exist an algorithm that achieves the quadratic lower bound derived here.
 
Recent work has also investigated Grover's algorithm from the point of view of post-quantum theories \cite{aaronson2016space,bao2015grover}. These works considered modifications of quantum theory which allow for superluminal signalling and cloning of states. In contrast, the generalised probabilistic theory framework employed here allowed us to investigate Grover's lower bound in alternate theories that are physically reasonable and which, for example, do not allow for superluminal signalling \cite{barrett2007information} or cloning \cite{barnum}. %Based on this, researchers interested in exploring post-quantum theories---such as those arising from the black hole firewall and information loss paradoxes---may find this framework an appealing playground in which to explore modifications of quantum theory. Indeed, preliminary work has already begun in this direction \cite{muller2012black}.

As theories satisfying our five physical principles appear `quantum-like'---at least from the point of view of the search problem---investigating interference behaviour in them may inform current experiments searching for post-quantum interference.

\emph{Acknowledgements---}The authors thank H. Barnum and M. J. Hoban for useful discussions and M. J. Hoban for proof reading a draft of the current paper. The authors also acknowledge encouragement and support from J. J. Barry. This work was supported by the EPSRC through the Controlled Quantum Dynamics Centre for Doctoral Training and the Oxford Department of Computer Science. CML also acknowledges funding from University College, Oxford. 

\appendix 
\section{Results for coherences} \label{AppendixCoherence}
\subsection{Proof of lemma~\ref{lemma: decompisition of the identity into coherences}} 
In a theory with maximal order of interference $\OInt$ one has
\[\mathds{1}_N = \sum_{{\small \begin{array}{c} I\subseteq \mathbf{N} \\ |I|\leq \OInt \end{array}}}\mathcal{C}\left(\OInt,|I|,N\right)P_I.\]
Thus, showing $\mathds{1}_N=\sum_{|I|=1}^\OInt \omega_I$ reduces to showing
\[\sum_{|I|=1}^\OInt \omega_I= \sum_{{\small \begin{array}{c} I\subseteq \mathbf{N} \\ |I|\leq \OInt \end{array}}}\mathcal{C}\left(\OInt,|I|,N\right)P_I.\] As $\omega_I:=\sum_{\tilde{I}\subseteq I}(-1)^{|I|+|\tilde{I}|}P_{\tilde{I}}$, we just have to count the number of $P_I$'s that appear as we sum over $|I|$. For some fixed $I$, this is just
$$ \sum_{\alpha=|I|}^\OInt (-1)^{\alpha-|I|} \left(\begin{array}{c}N-|I|\\\alpha-|I|\end{array}\right).$$ By expanding and rearranging this, one can straightforwardly (if tediously) show that this equals $\mathcal{C}\left(\OInt,|I|,N\right)$, and we are done.
\subsection{Proof of lemma~\ref{lemma: orthogonality of coherences} part $\mathrm{i)}$\label{AppendixCoherence1}}
From the definition of $\omega_I$, it follows that
\[\omega_I\omega_J = (-1)^{|I|+|J|}\sum_{\widetilde{I}\subseteq I}\sum_{\widetilde{J}\subseteq J}(-1)^{|\widetilde{I}|+|\widetilde{J}|}P_{\widetilde{I}}P_{\widetilde{J}}\]
\[=(-1)^{|I|+|J|}\sum_{\widetilde{K}\subseteq I\cap J}\mathcal{D}\left(I,J,\widetilde{K}\right) P_{\widetilde{K}}\]
where $\mathcal{D}\left(I,J,\widetilde{K}\right)$ is the number of distinct pairings of $\widetilde{I}$ and $\widetilde{J}$ such that $\widetilde{I}\cap\widetilde{J}=\widetilde{K}$ and $|\widetilde{I}|+|\widetilde{J}|$ is even, minus the number of distinct pairings where $\widetilde{I}\cap\widetilde{J}=\widetilde{K}$ and $|\widetilde{I}|+|\widetilde{J}|$ is odd. It will now be shown that
\[\mathcal{D}\left(I,J,\widetilde{K}\right)= \left\{ \begin{array}{cc} 0 & \text{ if } I\neq J \\ (-1)^{|I|+|\widetilde{K}|} & \text{ if } I=J \end{array}\right. \]
 
For the $I\neq J$ case fix some particular $i \in I$ such that $i\not\in J$ and consider some $\widetilde{I}\subseteq{I},\widetilde{J}\subseteq{J}$ such that $\widetilde{I}\cap\widetilde{J}=\widetilde{K}$. If $x \notin \widetilde{I}$ alter $\widetilde{I}$ by adding $i$, otherwise alter $\widetilde{I}$ by removing $x$. This procedure turns each even $|\widetilde{I}|+|\widetilde{J}|$, odd. We have thus shown that for each $\widetilde{I}\subseteq{I}$ and $\widetilde{J}\subseteq{J}$ such that $\widetilde{I}\cap\widetilde{J}=\widetilde{K}$ and $|\widetilde{I}|+|\widetilde{J}|$ is even, there exists an $\widetilde{I}'\subseteq{I}$ such that $\widetilde{I}'\cap\widetilde{J}=\widetilde{K}$ and $|\widetilde{I}'|+|\widetilde{J}|$ is odd, and vice versa. Thus the number of distinct pairings of $\widetilde{I}$ and $\widetilde{J}$ such that $\widetilde{I}\cap\widetilde{J}=\widetilde{K}$ and $|\widetilde{I}|+|\widetilde{J}|$ is even is equal to the number of distinct pairings of $\widetilde{I}$ and $\widetilde{J}$ such that $\widetilde{I}\cap\widetilde{J}=\widetilde{K}$ and $|\widetilde{I}|+|\widetilde{J}|$ is odd, and so $\mathcal{D}\left(I,J,\widetilde{K}\right)=0$ when $I\neq J$.

For the $I=J$ case we can make a similar argument by picking some $i\in I, i\not\in \widetilde{J}$ except for when $\widetilde{J}=J=I$. This case gives an excess $\pm 1$ depending on whether $|J|+|\widetilde{K}|$ is odd or even, implying $\mathcal{D}\left(I,J,\widetilde{K}\right)=(-1)^{|I|+|\widetilde{K}|}$ when $I=J$.

This immediately gives $\omega_I\omega_J=0$ if $I\neq J$ and,
\[\omega_I\omega_I = (-1)^{2|I|}\sum_{\widetilde{K}\subseteq I}(-1)^{|I|+|\widetilde{K}|} P_{\widetilde{K}} = \omega_I\]
if $I=J$.

\subsection{Proof of lemma~\ref{lemma: orthogonality of coherences} part $\mathrm{ii)}$\label{AppendixCoherence2}}
 
To prove the lemma, we need the fact that the $\omega_I$'s are self-dual $\omega_I^\dagger=\omega_I$, where the $\dagger$ is defined by the the self-dualising inner-product as: $\langle \cdot, \omega_I \cdot\rangle=\langle\omega_I^\dagger\cdot, \cdot\rangle$. Recalling that the $\omega_I$'s correspond to linear combinations of the $P_I$'s, this follows immediately from self-duality of the projectors $P_I$, which is proved in \cite{barnum2014higher} (Recall that principles $1$ to $5$ imply the first two axioms of \cite{barnum2014higher}). We now have
\[\begin{aligned}
\Vert s \Vert^2 &= \langle s, s\rangle = \langle\sum_I\omega_I s, \sum_J\omega_J s\rangle \end{aligned} \]
\[=\sum_{I,J}\langle\omega_I s,\omega_J s\rangle =\sum_{I,J}\langle s,\omega_I^\dagger\omega_J s\rangle\]
\[=\sum_{I,J}\langle s,\omega_I\omega_J s\rangle=\sum_{I,J}\delta_{IJ}\langle s,\omega_I s\rangle\]
where the last equality follows from the orthogonality of the $\omega_I$'s. Finally
\[\begin{aligned}\Vert s\Vert^2=\sum_{I}\langle s,\omega_I{s}\rangle=\sum_{I}\langle\omega_I{s},\omega_I{s}\rangle= \sum_I\lVert\omega_I{s}\rVert^2 
\end{aligned}\]

%\section{Main result}
%\subsection{Proof of $\sum_x \lVert(\mathds{1}-\mathcal{O}_x){s}\rVert^2 \leq \OInt$} \label{Appendix:SingleQueryBound}

\subsection{Proof of $D_k \geq cN$} \label{Appendix: lower bound for main proof}

We assume that $\langle x,{s}_k^x \rangle \geq1/2$ for all $x$, so a measurement of $s_k^x$ yields a solution to the search problem with probability at least $1/2$. Let $E_k=\sum_x \Vert {s}_k^x - x \Vert^2$ and $F_k=\sum_x \Vert {s}_k - x \Vert^2$. It follows that 
$$\begin{aligned} &\text{i) }E_k=\sum_x 2(1- \langle x,{s}_k^x \rangle ) \leq \sum_x 2(1- 1/2) \leq N\text{ and,} \\
 &\text{ii) } F_k\geq 2\left(N-\Vert {s}_k \Vert \sqrt{\left\langle \sum_x x, \sum_y y \right\rangle}\right)\geq 2\left(N-\sqrt{N}\right) \end{aligned}$$
where $\text{ii)}$ follows from the Cauchy-Schwarz inequality, $\Vert {s}_k \Vert \leq 1$ and $\langle x,y \rangle = \delta_{xy}$. As explicitly calculated on page $270$ of \cite{nielsen2010quantum}, by using the reverse triangle inequality and the Cauchy-Schwarz inequality, it follows that
$D_k \geq \left( \sqrt{F_k}-\sqrt{E_k}\right)^2$. Combing this with the upper bound on $E_k$ and the lower bound on $F_k$, we have that $D_k \geq cN,$ for sufficiently large $N$, where $c$ is any constant less than $\left(\sqrt{2}-1\right)^2\approx 0.17$.


\begin{thebibliography}{10} 

\bibitem{aaronson2016space}
Scott Aaronson, Adam Bouland, Joseph Fitzsimons, and Mitchell Lee.
\newblock The space just above bqp.
\newblock In {\em Proceedings of the 2016 ACM Conference on Innovations in
  Theoretical Computer Science}, pages 271--280. ACM, 2016.

\bibitem{bao2015grover}
Ning Bao, Adam Bouland, and Stephen~P Jordan.
\newblock Grover search and the no-signaling principle.
\newblock {\em arXiv preprint arXiv:1511.00657}, 2015.

\bibitem{barnum2014higher}
Howard Barnum, Markus~P M{\"u}ller, Cozmin Ududec, et~al.
\newblock Higher-order interference and single-system postulates characterizing
  quantum theory.
\newblock {\em New Journal of Physics}, 16(12):123029, 2014.

\bibitem{barrett2007information}
Jonathan Barrett.
\newblock Information processing in generalized probabilistic theories.
\newblock {\em Physical Review A}, 75(3):032304, 2007.

\bibitem{barrett2015landscape}
Jonathan Barrett, Niel de~Beaudrap, Ciar{\'a}n~M Lee, and Matty~J Hoban.
\newblock The computational landscape of general physical theories.
\newblock {\em Forthcoming}, 2016.

\bibitem{bennett1997strengths}
Charles~H Bennett, Ethan Bernstein, Gilles Brassard, and Umesh Vazirani.
\newblock Strengths and weaknesses of quantum computing.
\newblock {\em SIAM journal on Computing}, 26(5):1510--1523, 1997.

\bibitem{boyer1996tight}
Michel Boyer, Gilles Brassard, Peter H{\o}yer, and Alain Tapp.
\newblock Tight bounds on quantum searching.
\newblock {\em arXiv preprint quant-ph/9605034}, 1996.

\bibitem{Pavia1}
Giulio Chiribella, Giacomo~Mauro D'Ariano, and Paolo Perinotti.
\newblock Probabilistic theories with purification.
\newblock {\em Physical Review A}, 81(6):062348, 2010.

\bibitem{Pavia2}
Giulio Chiribella, Giacomo~Mauro D'Ariano, and Paolo Perinotti.
\newblock Informational derivation of quantum theory.
\newblock {\em Physical Review A}, 84(1):012311, 2011.

\bibitem{chiribella2015operational}
Giulio Chiribella and Carlo~Maria Scandolo.
\newblock Operational axioms for state diagonalization.
\newblock {\em arXiv preprint arXiv:1506.00380}, 2015.

\bibitem{dowker2014histories}
Fay Dowker, Joe Henson, and Petros Wallden.
\newblock A histories perspective on characterizing quantum non-locality.
\newblock {\em New Journal of Physics}, 16(3):033033, 2014.

\bibitem{grover1997quantum}
Lov~K Grover.
\newblock Quantum mechanics helps in searching for a needle in a haystack.
\newblock {\em Physical review letters}, 79(2):325, 1997.

\bibitem{Hardy2011}
Lucien Hardy.
\newblock Reformulating and reconstructing quantum theory.
\newblock {\em arXib preprint arXiv:1104.2066}, 2011.

\bibitem{henson2015bounding}
Joe Henson.
\newblock Bounding quantum contextuality with lack of third-order
  intereference.
\newblock {\em Phys. Rev. Lett. (arXiv:1406.3281)}, (114):220403, 2015.

\bibitem{janotta2013generalized}
Peter Janotta and Raymond Lal.
\newblock Generalized probabilistic theories without the no-restriction
  hypothesis.
\newblock {\em Physical Review A}, 87(5):052131, 2013.

\bibitem{lee2015computation}
Ciar{\'a}n~M Lee and Jonathan Barrett.
\newblock Computation in generalised probabilisitic theories.
\newblock {\em New Journal of Physics}, 17(8):083001, 2015.

\bibitem{lee2015proofs}
Ciar{\'a}n~M Lee and Matty~J Hoban.
\newblock Proofs and advice in general physical theories: a trade-off between
  states and dynamics?
\newblock {\em arXiv preprint arXiv:1510.04702}, 2015.

\bibitem{lee2015generalised}
Ciar{\'a}n~M Lee and John~H Selby.
\newblock Generalised phase kick-back: a computational advantage for
  higher-order interference?
\newblock {\em arXiv preprint arXiv:1510.04699}, 2015.

\bibitem{lee2015higher}
Ciar{\'a}n~M Lee and John~H Selby.
\newblock Higher-order interference in extensions of quantum theory.
\newblock {\em arXiv preprint arXiv:1510.03860}, 2015.

\bibitem{muller2012structure}
Markus~P M{\"u}ller and Cozmin Ududec.
\newblock Structure of reversible computation determines the self-duality of
  quantum theory.
\newblock {\em Physical Review Letters}, 108(13):130401, 2012.

\bibitem{navascues2015almost}
Miguel Navascu{\'e}s, Yelena Guryanova, Matty~J Hoban, and Antonio Ac{\'\i}n.
\newblock Almost quantum correlations.
\newblock {\em Nature communications}, 6, 2015.

\bibitem{nielsen2010quantum}
Michael~A Nielsen and Isaac~L Chuang.
\newblock {\em Quantum computation and quantum information}.
\newblock Cambridge university press, 2010.

\bibitem{niestegge2012conditional}
Gerd Niestegge.
\newblock Conditional probability, three-slit experiments, and the jordan
  algebra structure of quantum mechanics.
\newblock {\em Advances in Mathematical Physics}, 2012, 2012.

\bibitem{park2012three}
Daniel~K Park, Osama Moussa, and Raymond Laflamme.
\newblock Three path interference using nuclear magnetic resonance: a test of
  the consistency of born's rule.
\newblock {\em New Journal of Physics}, 14(11):113025, 2012.

\bibitem{popescu1998causality}
Sandu Popescu and Daniel Rohrlich.
\newblock {\em Causality and nonlocality as axioms for quantum mechanics}.
\newblock Springer, 1998.

\bibitem{sinha2015superposition}
Aninda Sinha, Aravind~H Vijay, and Urbasi Sinha.
\newblock On the superposition principle in interference experiments.
\newblock {\em Scientific reports}, 5, 2015.

\bibitem{sinha2010ruling}
Urbasi Sinha, Christophe Couteau, Thomas Jennewein, Raymond Laflamme, and
  Gregor Weihs.
\newblock Ruling out multi-order interference in quantum mechanics.
\newblock {\em Science}, 329(5990):418--421, 2010.

\bibitem{sinha2008testing}
Urbasi Sinha, Christophe Couteau, Zachari Medendorp, Immo S{\"o}llner, Raymond
  Laflamme, Rafael Sorkin, and Gregor Weihs.
\newblock Testing born's rule in quantum mechanics with a triple slit
  experiment.
\newblock {\em arXiv preprint arXiv:0811.2068}, 2008.

\bibitem{sorkin1994quantum}
Rafael~D Sorkin.
\newblock Quantum mechanics as quantum measure theory.
\newblock {\em Modern Physics Letters A}, 9(33):3119--3127, 1994.

\bibitem{sorkin1995quantum}
Rafael~D Sorkin.
\newblock Quantum measure theory and its interpretation.
\newblock {\em arXiv preprint gr-qc/9507057}, 1995.

\bibitem{spekkens2007evidence}
Robert~W Spekkens.
\newblock Evidence for the epistemic view of quantum states: A toy theory.
\newblock {\em Physical Review A}, 75(3):032110, 2007.

\bibitem{stahlke2014quantum}
Dan Stahlke.
\newblock Quantum interference as a resource for quantum speedup.
\newblock {\em Physical Review A}, 90(2):022302, 2014.

\bibitem{ududec2011three}
Cozmin Ududec, Howard Barnum, and Joseph Emerson.
\newblock Three slit experiments and the structure of quantum theory.
\newblock {\em Foundations of Physics}, 41(3):396--405, 2011.

\bibitem{zalka1999grover}
Christof Zalka.
\newblock Grover’s quantum searching algorithm is optimal.
\newblock {\em Physical Review A}, 60(4):2746, 1999.

\bibitem{barnum}
Howard Barnum, Jonathan Barrett, Matthew Leifer, and Alexander Wilce.
\newblock A generalized no-broadcasting theorem
\newblock {\em Phys. Rev. Lett.}, 99:240501, 2007.


\end{thebibliography}
\end{document}